\documentclass[11pt]{article}

\usepackage{amsmath,amssymb,amsthm,mathtools}
\usepackage[margin=1in]{geometry}
\usepackage{enumitem}
\usepackage{bm}
\usepackage{microtype}
\usepackage{graphicx}
\usepackage{xcolor}
\usepackage{booktabs}
\usepackage{threeparttable}
\usepackage{makecell}
\usepackage{multirow}
\usepackage{float}
\usepackage[ruled,vlined,linesnumbered,noend]{algorithm2e}
\SetKw{KwRet}{return}
\SetKw{KwBreak}{break}

\SetCommentSty{mycommfont}
\usepackage{tikz}
\usetikzlibrary{arrows.meta,positioning,calc}
\usepackage{natbib}
\usepackage{hyperref}
\usepackage[nameinlink,capitalize]{cleveref}

\theoremstyle{plain}
\newtheorem{theorem}{Theorem}[section]
\newtheorem{lemma}[theorem]{Lemma}

\theoremstyle{definition}

\newtheorem{assumption}[theorem]{Assumption}

\newcommand{\GFT}{\operatorname{GFT}}
\newcommand{\Rev}{\operatorname{Rev}}
\newcommand{\SB}{\operatorname{SB}}
\newcommand{\FB}{\operatorname{FB}}
\newcommand{\Var}{\operatorname{Var}}
\newcommand{\supp}{\operatorname{supp}}

\newcommand{\KL}{\operatorname{KL}}
\newcommand{\TV}{\operatorname{TV}}
\newcommand{\E}{\mathbb{E}}

\newcommand{\Prob}{\mathbb{P}}
\newcommand{\one}{\mathbf{1}}
\newcommand{\eps}{\varepsilon}
\newcommand{\Db}{\widehat{D}}

\newcommand{\Ext}{\operatorname{Ext}}

\newcommand{\Oe}{\widetilde{O}}

\newcommand{\xhdr}[1]{\vspace{2mm} \noindent{\bf #1}}

\usepackage{color-edits}
\addauthor[Zongqi]{zq}{blue}
\addauthor[Shengxin]{shengxin}{orange}
\addauthor[Qiaoyun]{qy}{purple}

\title{Sample Complexity of the Second-Best Bilateral Trade}
\author{
    Qiaoyun Shi\thanks{Harbin Institute of Technology, Shenzhen; Great Bay University, \href{mailto:qyunyun26.s@gmail.com}{\texttt{qyunyun26.s@gmail.com}}}
    \and
    Shengxin Liu\thanks{Harbin Institute of Technology, Shenzhen, \href{mailto:sxliu@hit.edu.cn}{\texttt{sxliu@hit.edu.cn}}}
    \and
    Zongqi Wan\thanks{Great Bay University, \href{mailto:zqwan@gbu.edu.cn}{\texttt{zqwan@gbu.edu.cn}}}
}
\date{}

\begin{document}
\setcounter{page}{0}
\setcounter{tocdepth}{1}

\maketitle

\begin{abstract}
We study the sample complexity of learning near-optimal bilateral trade mechanisms.  Unlike previous work on learning simple or fixed-price bilateral-trade mechanisms~\citep{BabaioffFreyNisan2024,CastiglioniLunghiMarchesi2026}, we focus on mechanisms satisfying Bayesian incentive compatible (BIC), interim individually rational (IIR), and ex-ante weakly budget balanced (WBB). In other words, our target is to design a sample-based mechanism that achieves the \emph{second-best} gains-from-trade benchmark.

We give matching or nearly matching upper and lower bounds in three regimes.
For regular product distributions on $[0,h]^2$, additive $\eps$-approximation has sample complexity $\widetilde{\Theta}(h^2/\eps^2)$. For multiplicative $(1-\alpha)$-approximation on distributions under the same assumptions, we find that the sample complexity is $\widetilde{\Theta}(\frac{h}{\SB(D)\cdot \alpha^2})$, which is benchmark-sensitive with unavoidable dependence on the second-best gains-from-trade $\SB(D)$.  We also investigate the setting of unbounded distributions with monotone-hazard-rate assumption.  The sample complexity depends on the ratio $\chi_\mu(D)=\mu(D)/\SB(D)$, where $\mu(D)$ is the sum of the buyer's expected value and seller's expected cost.
\end{abstract}


\newpage


\section{Introduction}\label{sec:main-intro}

Bilateral trade is perhaps the smallest market in which the basic tension of two-sided market design already appears.
In this model, a seller owns one indivisible item and has a private cost $s$ drawn from a distribution $D_S$, while a buyer has a private value $b$ drawn independently from $D_B$.  If the item is transferred from the seller to the buyer, the realized gains-from-trade (GFT) is $b-s$: this is the increase in total surplus created by moving the item to the agent who values it more.  If no trade occurs, the GFT is zero.  Thus the efficient, full-information rule trades exactly when $b>s$, and under the product distribution $D=D_B\times D_S$ its expected GFT is
$
\FB(D)=\E_{(b,s)\sim D}\bigl[(b-s)^+\bigr],
$
which is also called the first-best GFT.

However, private information prevents the designer from directly implementing the efficient rule.  The mechanism must incentivize the buyer and seller to report their value and cost truthfully, guarantee nonnegative interim utility, and avoid running an expected deficit when buying the good from the seller and selling it to the buyer. 
The seminal result of \cite{MyersonSatterthwaite1983} shows that any mechanism satisfying the above constraints, namely Bayesian incentive compatible (BIC),
interim individually rational (IIR), and ex-ante weakly budget balanced (ex-ante WBB), can never achieve the first-best GFT.
Instead, they proposed the \emph{second-best} GFT, denoted by $\SB(D)$, which is the largest expected GFT achievable by a BIC, IIR, and ex-ante WBB mechanism. This benchmark is central to the theory of bilateral trade and two-sided markets, and has motivated a large literature on simple mechanisms and approximation guarantees~\citep{BCWZ2017,DMSW2022,Fei2022,CaiWu2023,LiuRenWang2023}.

Most of the classical bilateral trade literature assumes that the designer knows the buyer and seller distributions. In many applications, such as online platforms, the prior is not given explicitly. The designer may instead observe past buyer values and seller costs, or have access to samples from the two sides of the market.  This leads to a learning question: how many samples are needed before the designer can choose a mechanism whose GFT is close to the relevant benchmark?

Recent work has begun to answer this question for simple mechanisms in small markets.  In particular, \citet{BabaioffFreyNisan2024} study learning simple mechanisms that maximize GFT from samples, and \citet{CastiglioniLunghiMarchesi2026} study sample complexity for fixed-price mechanisms.  These works show that GFT learning already has its own structure even for simple trading rules. 
This broader question is explicitly raised by \citet{BabaioffFreyNisan2024}, who ask about learning BIC/IIR/WBB mechanisms beyond fixed prices and simple strongly budget-balanced mechanisms.

We take up this question for bilateral trade.  Given independent samples from $D_B$ and $D_S$, the learner must output a BIC, IIR, and ex-ante WBB mechanism whose GFT is close to the second-best GFT under the unknown true distribution.  The central question is:
\begin{quote}
\emph{How many samples are needed to learn a nearly second-best bilateral-trade mechanism?}
\end{quote}

To answer this question, we give matching or nearly matching upper and lower bounds in three regimes: additive approximation on bounded-support distributions, multiplicative approximation on bounded-support distributions, and multiplicative approximation for unbounded monotone hazard rate (MHR) distributions.  
In particular, the multiplicative approximation is benchmark-sensitive: when the second-best GFT is small, no sample bound independent of $\SB(D)$ is possible.

\subsection{Our Results and Techniques}

Table~\ref{tab:main-results} summarizes the three regimes considered in this paper. Before introducing these results in detail, we highlight two conceptual and technical takeaways. 

\xhdr{Distribution-relative feasibility and mechanism mixing.}
The first takeaway is that feasibility is distribution-relative. A mechanism with nonnegative expected revenue under the empirical product distribution may run a deficit under the true distribution, and a mechanism feasible on one statistically indistinguishable instance may be infeasible on another. Our upper bounds therefore impose a positive empirical revenue margin large enough to absorb the sample-to-truth slack. To show that this margin preserves near-optimality, we use a posted-spread mechanism to construct a near-optimal comparator satisfying the margin. The learned mechanism then satisfies the empirical revenue margin, which is sufficient to remain WBB after transfer to the true distribution.

\xhdr{Benchmark-sensitive multiplicative approximation.}
The second takeaway concerns multiplicative approximation.  In the Myerson auction learning literature~\citep{DHP2016, GuoHuangZhang2019}, multiplicative approximation guarantees are stated without dependence on the realized optimal revenue.  However, such a scale-free guarantee is impossible in our bilateral trade setting. Our lower bounds exploit the same distribution-relative feasibility from the opposite direction: we first construct statistically close instances whose near-optimal feasible mechanism classes are separated by the ex-ante WBB constraint. Under multiplicative approximation, this separation can be diluted into a rare region without weakening it relative to the benchmark, while proportionally reducing the statistical information available from each sample. This is the source of the $h/\SB(D)$ dependence.

\begin{table}[t]
    \centering
    \caption{
        A summary of our sample-complexity bounds for learning second-best gains from trade, up to logarithmic factors.}
    \resizebox{\textwidth}{!}
    {\begin{tabular}{lccc}
\toprule
\textbf{Settings} & \textbf{Guarantee} & \textbf{Upper Bound} & \textbf{Lower Bound}
\\
\midrule
\textbf{\makecell[l]{Bounded support $[0,h]^2$, regular}} &
\makecell{$\SB(D)-\eps$\\$(\eps\le c\SB(D))$} &
\makecell{$\widetilde O\!\left(\frac{h^2}{\eps^2}\right)$\\{[}Thm.~\ref{thm:add-upper-main}{]}} &
\makecell{$\Omega\!\left(\frac{h^2}{\eps^2}\right)$\\{[}Thm.~\ref{thm:add-lower-main}{]}}
\\
\cmidrule(lr){1-4}
\textbf{\makecell[l]{Bounded support $[0,h]^2$, regular}} &
\makecell{$(1-\alpha)\SB(D)$} &
\makecell{$\widetilde O\!\left(
\frac{h}{\SB(D)}\frac{1}{\alpha^2}
\right)$\\{[}Thm.~\ref{thm:bounded-mult-upper-main}{]}} &
\makecell{$\Omega\!\left(
\frac{h}{\SB(D)}\frac{1}{\alpha^2}
\right)$\\{[}Thm.~\ref{thm:bounded-mult-lower-main}{]}}
\\
\cmidrule(lr){1-4}
\textbf{\makecell[l]{MHR marginals, seller regularity}} &
\makecell{$(1-\alpha)\SB(D)$} &
\makecell{$\widetilde O\!\left(
\frac{\mu(D)}{\SB(D)}\frac{1}{\alpha^2}
\right)$\\{[}Thm.~\ref{thm:mhr-upper-main}{]}} &
\makecell{$\widetilde\Omega\!\left(
\frac{\mu(D)}{\SB(D)}\frac{1}{\alpha^2}
\right)$\\{[}Thm.~\ref{thm:mhr-lower-main}{]}}
\\
\bottomrule
\end{tabular}
}
    \begin{tablenotes}
    \item \emph{\footnotesize
    \underline{Note}: Upper bounds hold with probability at least $1-\gamma$ and use independent samples from each marginal. Precise bounds are stated in the referenced theorems. Here $\mu(D)=\E[B]+\E[S]$.}
    \end{tablenotes}
\label{tab:main-results}
\end{table}

Next, we introduce our results.

\xhdr{Bounded support: additive approximation.} 
The upper bound follows a simple robust-learning principle: optimize over an empirical version of the canonical mechanism class with a sufficient revenue margin, and use a posted-spread mechanism to construct a near-optimal comparator satisfying that margin.  This gives a BIC, IIR, and ex-ante WBB mechanism with
$\GFT(M;D)\ge \SB(D)-\eps$
using
\[
\widetilde O\!\left(\frac{h^2}{\eps^2}\right)
\]
samples from each marginal.
The additive lower bound matches the leading $h^2/\eps^2$ scale.  We construct two statistically close regular instances whose feasible near-optimal mechanism classes are separated by the ex-ante WBB constraint: any mechanism that is WBB on one instance incurs an $\Omega(\eps)$ GFT loss on the other. Distinguishing which feasible class is appropriate therefore requires $\Omega(h^2/\eps^2)$ samples.

\xhdr{Bounded support: multiplicative approximation.}
For the multiplicative guarantee $\GFT(M;D)\ge(1-\alpha)\SB(D)$, substituting
$\eps=\alpha\SB(D)$ into the additive result would give a quadratic dependence on
$h/\SB(D)$.
We improve this to the nearly tight characterization
\[
\widetilde\Theta\!\left(
\frac{h}{\SB(D)}\frac1{\alpha^2}
\right).
\]
The algorithm differs from the additive learner by first estimating the scale of $\SB(D)$ and then learning only among mechanisms whose empirical accounting quantities live at that scale.  This localization turns the additive $h^2/\eps^2$-type dependence into the sharper $h/\SB(D)$ dependence.

The lower bound shows that this benchmark sensitivity is inherent to multiplicative approximation. Starting from the additive hard pair, we dilute the distinguishing part of the instance into a block that occurs with probability $p$. This scales both the second-best benchmark and the cross-instance WBB separation by the same factor $p$, so the separation remains unchanged relative to the benchmark, while the information revealed by each sample is reduced by a factor $p$. Since the resulting instances satisfy $\SB(D)=\Theta(ph)$, distinguishing the appropriate feasible mechanism incurs an additional $1/p$ sample cost, yielding the linear dependence on $h/\SB(D)$.

\xhdr{Unbounded MHR distributions: multiplicative approximation.}
Finally, we consider the multiplicative approximation without a bounded support. The main new issue is that there is no known scale $h$ on which to run the bounded-support learner. Under MHR marginals with seller regularity,\footnote{MHR implies buyer-side regularity, but does not imply seller-side regularity.} we use the samples to identify a sufficiently large effective window, learn on the induced bounded problem, and then extend the learned mechanism back to the original type space.
The resulting sample complexity depends on the mean-to-benchmark ratio
\[
\chi_\mu(D)=\frac{\E[B]+\E[S]}{\SB(D)}.
\]
The corresponding sample-complexity characterization is
\[
\widetilde\Theta\!\left(
\frac{\chi_\mu(D)}{\alpha^2}
\right).
\]
The lower bound reuses the same rare-block construction inside the upper tail of an MHR distribution. The rare block preserves the multiplicative WBB separation, while the surrounding exponential tail makes the mean-to-benchmark ratio scale as $\chi_\mu(D)$, yielding the matching dependence up to logarithmic factors.

\subsection{Related Work}

\xhdr{Bilateral trade and approximation algorithms.}
In bilateral trade with independent buyer and seller values, \citet{MyersonSatterthwaite1983} rule out first-best GFT under incentive compatibility, individual rationality, and budget balance, motivating Bayesian approximation of the second-best GFT benchmark, with early work by \citet{BlumrosenMizrahi2016}. \citet{BCWZ2017} develop a virtual-surplus and simple-mechanism framework, and \citet{Fei2022}, building on \citet{DMSW2022}, gives a delegated-pricing bound for the second-best benchmark. Fixed-price approximations form a closely related line of work~\citep{McAfee2008,KangPerniceVondrak2022,LiuRenWang2023,CaiWu2023}. Recent work further sharpens the constant-factor guarantees: \citet{LQRW2026} establish the sharp first-best/second-best ratio, while \citet{HartlineWang2025} give a simple geometric analysis of constant-factor GFT approximation. \citet{DuttingFuscoLazosLeonardiReiffenhauser2021} study limited-information two-sided mechanisms under different information models and benchmarks.  Correlated bilateral trade has been studied separately as a distinct extension of the independent-values model. \citet{DobzinskiShaulker2024} study welfare approximation under correlated values, while related interdependent-value models are considered by \citet{DobzinskiEdenGoldnerShaulkerTsilivis2025,KunimotoZhang2026}.

\xhdr{Learning of bilateral trade.}
The closest learning papers study sample-based GFT maximization and fixed-price learning in small two-sided markets~\citep{BabaioffFreyNisan2024,CastiglioniLunghiMarchesi2026}. \citet{DengMaoSivanWangWu2025} study sample-based first-best approximation in bilateral trade. Our target is instead the second-best benchmark under BIC, IIR, and ex-ante WBB constraints for the true distribution. \citet{DiGregorioFuscoLeonardiSchwiegelshohn2026} study private sample-based learning in bilateral trade.  A separate line studies online bilateral trade under regret objectives and different budget-balance and feedback models~\citep{BlumSandholmZinkevich2006, CesaBianchiEtAl2024a,CesaBianchiEtAl2024b,AzarFiatFusco2024, BernasconiCastiglioniCelliFusco2024,ChenJinLuZhang2025,Jin2026, LunghiCastiglioniMarchesi2026,GaucherBernasconiCastiglioniCelliPerchet2025}. For GFT maximization, much of this literature studies fixed-price mechanisms and related price-based benchmarks. \citet{BernasconiCastiglioniCelliFusco2024} consider global budget balance, \citet{ChenJinLuZhang2025} establish tight regret bounds for fixed-price bilateral trade, and \citet{LunghiPiccinatoCastiglioniMarchesi2026} study a stronger benchmark based on distributions over pairs of prices.  Feng, Ma, Peng, and Wan~\citep{FengMaPengWan2026} study online pricing in two-sided markets.\citet{DiGregorioDuttingFuscoSchwiegelshohn2025} study profit maximization over general DSIC and IR mechanisms. These works differ from our setting in the mechanism classes, objectives, or learning model considered.


\section{Preliminaries}\label{sec:main-model}
For a product distribution $D=D_B\times D_S$, let $M=(q,x_B,x_S)$ be a direct mechanism, where $q$ denotes the trade probability, $x_B$ and $x_S$ denote the buyer and seller payments, respectively. Define the GFT and revenue of $M$ as
\[
\GFT(M;D)=\E_D[q(b,s)(b-s)],\qquad
\Rev(M;D)=\E_D[x_B(b,s)-x_S(b,s)].
\]
A mechanism is feasible if it is BIC, IIR, and ex-ante WBB under $D$, where ex-ante WBB means $\Rev(M;D)\ge0$.  
Given the type distribution $D$, the second-best GFT is defined as 
\[
\SB(D)=\sup\{\GFT(M;D):M\text{ is feasible on }D\},
\]
and the first-best GFT is $\FB(D)=\E_{(b,s)\sim D}[(b-s)^+]$.
The first-best and second-best benchmarks are within a universal constant factor.
\begin{lemma}[\citet{Fei2022}]\label{lem:scale-main}
$
\SB(D)\le \FB(D)\le 3.15\SB(D).
$
\end{lemma}

\xhdr{Regularity and monotone hazard rate.}
Let $F_B,F_S$ be the buyer and seller marginal cdfs, with densities $f_B,f_S$ on their support
interiors.  The buyer virtual value and seller virtual cost are
\[
\phi_B(b)=b-\frac{1-F_B(b)}{f_B(b)},\qquad
\phi_S(s)=s+\frac{F_S(s)}{f_S(s)}.
\]
A buyer distribution is regular if its virtual value is nondecreasing, and a seller distribution is regular if its virtual cost is nondecreasing. 
For a distribution $F$ with density $f$, the \emph{hazard rate} is $h_F(x)\coloneqq f(x)/(1-F(x))$.  
A distribution $F$ is \emph{monotone hazard rate (MHR)} if $h_F(x)$ is nondecreasing on the support interior~\citep{BMP1963}. 
The virtual values are useful because the Myerson envelope formula bounds revenue by virtual surplus.
\begin{lemma}[\citet{BCWZ2017}]\label{lem:vs-ineq-main}
For continuous marginals with positive densities, every BIC/IIR mechanism with allocation $q$
satisfies
\[
\Rev(M;D)\le \E_D[q(b,s)(\phi_B(b)-\phi_S(s))].
\]
\end{lemma}

\xhdr{Threshold mechanisms.}
A deterministic canonical threshold mechanism is specified by monotone threshold maps $\tau_B(s)$ and $\tau_S(b)$ and trades on $b\ge\tau_B(s)$, equivalently $s\le\tau_S(b)$, with threshold payments.  The boundary thresholds $+\infty$ and $-\infty$ encode the never-trade and always-trade conventions.  A \emph{randomized threshold mechanism} is a report-independent ex-ante randomization over finitely many deterministic canonical threshold mechanisms.  Such randomization preserves BIC and IIR, while GFT and revenue are affine.  Appendix~\ref{sec:prelim} fixes the boundary conventions used when thresholds meet atoms.
The next theorem shows that canonical threshold mechanisms suffice to attain the second-best GFT for bounded regular product distributions.
\begin{theorem}[\cite{MyersonSatterthwaite1983,BCWZ2017}]\label{fact:structure-main}
For every regular product distribution $D$ on $[0,h]^2$ with positive densities and $\SB(D)>0$, there
exists a second-best optimal feasible mechanism $M^*$ with $\Rev(M^*;D)\ge0$ whose allocation is
\[
q^*(b,s)=\one\{(b-s)+\lambda^*(\phi_B(b)-\phi_S(s))\ge0\}
\]
for some multiplier $\lambda^*\ge0$.
\end{theorem}

\xhdr{Learning model.}
All learning guarantees below are with respect to independent marginal samples $b_1,\ldots,b_m\sim D_B$ and $s_1,\ldots,s_m\sim D_S$.  With probability at least $1-\gamma$, the output mechanism $\widehat M$ should be 
BIC, IIR, and ex-ante WBB under the true distribution $D$. Our target is
\[ \GFT(\widehat M;D)\ge \SB(D)-\eps \qquad\text{or}\qquad \GFT(\widehat M;D)\ge (1-\alpha)\SB(D), \]
in the additive or multiplicative regime, respectively.

\section{Additive Approximation with Bounded Support Distributions}\label{sec:main-additive}

\subsection{Upper Bound}\label{ssec:add-upper-main}

This subsection gives the bounded-support additive learner.  The algorithm is a plug-in empirical optimizer, but with one deliberate modification: instead of merely requiring nonnegative empirical revenue, it requires a positive empirical revenue margin.  The margin lets empirical revenue feasibility transfer back to the unknown distribution.

Let $D=D_B\times D_S$ be the unknown product distribution on $[0,h]^2$.  Write the independent marginal samples as $b_1,\ldots,b_m$ and $s_1,\ldots,s_m$, and set the empirical distributions to be
\[
\widehat D_B=\frac1m\sum_{r=1}^m\delta_{b_r},
\qquad
\widehat D_S=\frac1m\sum_{r=1}^m\delta_{s_r},
\qquad
\widehat D=\widehat D_B\times\widehat D_S,
\]
where $\delta_x$ denotes the point mass at $x$.
Define
\[ \Lambda_\eps:=1+\log\frac h\eps,\qquad \rho_{\rm add}:=c_{\rm add}\frac{\eps}{\Lambda_\eps},\qquad \delta_{\rm K}:=c_{\rm K}\frac{\eps}{h\Lambda_\eps}, \]
where $c_{\rm add}>0$ is a sufficiently small universal constant and $c_{\rm K}>0$ is chosen small relative to $c_{\rm add}$.

\xhdr{The empirical learner.}
The learner first forms the endpoint-augmented empirical grids.  After merging duplicate sample values, write
\[
\begin{aligned}
B&:=\mathcal B_{\rm emp}=\{0,h,b_1,\ldots,b_m\}=\{\beta_0<\cdots<\beta_K\},\\
S&:=\mathcal S_{\rm emp}=\{0,h,s_1,\ldots,s_m\}=\{\sigma_0<\cdots<\sigma_L\}.
\end{aligned}
\]
Let $\widehat p_i^B=\widehat D_B(\{\beta_i\})$ and $\widehat p_j^S=\widehat D_S(\{\sigma_j\})$.  

We optimize GFT over direct mechanisms via an LP defined on the empirical grid $B\times S$; its solutions are then extended to $[0,h]^2$ using threshold maps.  The LP variables are allocation probabilities
$
q_{ij}:=q(\beta_i,\sigma_j)\in[0,1],
 0\le i\le K,  0\le j\le L.
$ 
We require the allocation to be monotone in both arguments.  Equivalently, the LP optimizes over
\[
\mathcal Q_{B,S}:=
\{q\in[0,1]^{(K+1)(L+1)}: q_{i+1,j}\ge q_{ij},\ q_{i,j}\ge q_{i,j+1}\},
\]
where the first inequality ranges over $i<K$ and all $j$, and the second over all $i$ and $j<L$.  When $q$ is integral (i.e. $\{0,1\}$), these constraints give a deterministic canonical threshold allocation on the grid, consistent with the threshold maps $\tau_B,\tau_S$ from Section~\ref{sec:main-model}.

\begin{algorithm}[htbp]
\caption{Empirical additive learner}\label{alg:add-upper-main}
\KwIn{$h,\eps$, samples $b_1,\ldots,b_m$ and $s_1,\ldots,s_m$}
\KwOut{A randomized threshold mechanism on $[0,h]^2$, or no trade}
Build $\widehat D_B,\widehat D_S$, the grids $B,S$, and masses $\widehat p_i^B,\widehat p_j^S$\;
Form $\mathcal Q_{B,S}$ and encode $x_B(q),x_S(q)$, $\widehat{\GFT}(q)$, and $\widehat{\Rev}(q)$ as linear functions of $q$\;
Solve the empirical LP~\eqref{eq:additive-main-lp}\;
\If{the LP is infeasible}{
    \Return no trade\;
}
Let $q^\star$ be an optimal solution and write $q^\star=\sum_{r\le 2}\lambda_r q^r$ with integral $q^r\in\mathcal Q_{B,S}$\;
\ForEach{component $r$}{
    Extend $q^r$ to $\widetilde q^r(b,s)=q^r(\pi_B(b),\pi_S(s))$\;
    Equip $\widetilde q^r$ with its threshold payments\;
}
\Return the report-independent mixture choosing component $r$ with probability $\lambda_r$\;
\end{algorithm}

For a monotone allocation $q$, we use the following Myerson payments to keep the mechanism BIC and IIR on the grid:
\[
\begin{aligned}
x_B(\beta_i,\sigma_j)
&=\beta_iq_{ij}
-\sum_{t=0}^{i-1}(\beta_{t+1}-\beta_t)q_{tj},
\qquad
x_S(\beta_i,\sigma_j)
&=\sigma_jq_{ij}
+\sum_{t=j+1}^{L}(\sigma_t-\sigma_{t-1})q_{it}.
\end{aligned}
\]
This normalization gives zero utility to the lowest buyer type and the highest seller type.  Hence the monotone allocation $q$ uniquely determines the direct mechanism $M=(q,x_B,x_S)$ on the grid, and its feasibility is shown in Lemma~\ref{lem:fin-impl-main}.  The empirical objective and revenue are therefore linear functions of $q$:
\[
\widehat{\GFT}(q)=\sum_{i,j}\widehat p_i^B\widehat p_j^S(\beta_i-\sigma_j)q_{ij},\qquad \widehat{\Rev}(q)=\sum_{i,j}\widehat p_i^B\widehat p_j^S\bigl(x_B(\beta_i,\sigma_j)-x_S(\beta_i,\sigma_j)\bigr).
\]
The empirical program is the finite LP
\begin{equation}\label{eq:additive-main-lp}
\max_{q\in\mathcal Q_{B,S}}\ \widehat{\GFT}(q)
\qquad\text{s.t.}\qquad
\widehat{\Rev}(q)\ge \rho_{\rm add}.
\end{equation}
It has at most $O(m^2)$ allocation variables and $O(m^2)$ monotonicity constraints, and thus can be solved in polynomial time.  The LP solution is then converted into a mechanism on the original report space $[0,h]^2$ as follows.  By Lemma~\ref{lem:lp}, we can take an optimal solution $q^\star$ and write it as a convex combination $q^\star=\sum_{r\le2}\lambda_r q^r$ of at most two monotone integral grid allocations $q^r\in\{0,1\}^{B\times S}$.  For each component, define the stepwise allocation $\widetilde q^r(b,s):=q^r(\pi_B(b),\pi_S(s))$,
where $\pi_B(b):=\max\{\beta_i:\beta_i\le b\}$ and $\pi_S(s):=\min\{\sigma_j:\sigma_j\ge s\}$. 
Lemma~\ref{lem:grid-step-extension} gives threshold maps $\tau_B^r,\tau_S^r$ satisfying
\[
\one\{b\ge\tau_B^r(s)\}=\widetilde q^r(b,s)=\one\{s\le\tau_S^r(b)\}.
\]
The returned mechanism first draws $r$ with probability $\lambda_r$, independently of reports, and then runs this deterministic threshold mechanism on the reported $(b,s)$, with the corresponding threshold payments.  Thus the object returned by the algorithm is a randomized threshold mechanism on $[0,h]^2$; on empirical profiles its payments agree with the envelope payments above, so its empirical GFT and revenue are exactly the LP objective and constraint values.  The construction uses only the samples, and the analysis below shows that the resulting mechanism is feasible and near-optimal under the unknown distribution $D$.

\xhdr{Sample complexity analysis.}
We now analyze why the margin-constrained empirical LP is feasible and near-optimal on the high-probability sampling event.  This is an existence argument: we first show that the empirical LP contains a near-optimal feasible mechanism, and then use empirical optimality and transfer to obtain the guarantee for the learned mechanism.  Let $F_B,F_S$ be the marginal cdfs of $D_B,D_S$, and $\widehat F_B,\widehat F_S$ be the empirical marginal cdfs.  Define the marginal Kolmogorov distance by
\[
\Delta_{\rm K}(\widehat D,D):=
\max\{\|\widehat F_B-F_B\|_\infty,\|\widehat F_S-F_S\|_\infty\}.
\]
The following lemma transfers marginal Kolmogorov distance to GFT and revenue difference.

\begin{lemma}\label{lem:add-transfer-main}
If $\Delta_{\rm K}(\widehat D,D)\le\delta$, then every randomized canonical threshold mechanism $M$ satisfies
\[ |\GFT(M;D)-\GFT(M;\widehat D)| + |\Rev(M;D)-\Rev(M;\widehat D)| =O(h\delta). \]
\end{lemma}

To compare the true second-best mechanism with the empirical LP, we use the following projection notation.  For a randomized canonical threshold mechanism $M$, let $\Pi_{\widehat D}M$ denote its empirical-grid projection: restrict each deterministic component's allocation to the grid $B\times S$ and implement the resulting finite-grid allocation using the Myerson payments (Lemma~\ref{lem:fin-impl-main}) to ensure BIC/IIR.

\begin{lemma}\label{lem:projection}
Let $M$ be a deterministic canonical threshold mechanism on $[0,h]^2$, and let $M_{\rm emp}:=\Pi_{\widehat D}M$.  Then
\[
\GFT(M_{\rm emp};\widehat D)=\GFT(M;\widehat D),
\qquad
\Rev(M_{\rm emp};\widehat D)\ge\Rev(M;\widehat D).
\]
The same holds for randomized threshold mechanisms.
\end{lemma}

On the event $\Delta_{\rm K}(\widehat D,D)\le\delta_{\rm K}$, a canonical optimizer $M^*$ for $D$ need not satisfy the empirical revenue margin after projection.  Recall that, by Theorem~\ref{fact:structure-main}, there is a second-best optimal mechanism $M^*$ with $\Rev(M^*;D)\ge 0 $ whose allocation is a canonical threshold rule.  For the empirical-grid projection of this optimum, $M^*_{\rm emp}:=\Pi_{\widehat D}M^*$,
\[
\GFT(M^*_{\rm emp};\widehat D) = \GFT(M^*;\widehat D) \ge \SB(D)-O(h\delta_{\rm K}),
\qquad
\Rev(M^*_{\rm emp};\widehat D)\ge \Rev(M^*;\widehat D) \ge -O(h\delta_{\rm K}).
\]
Thus $M^*_{\rm emp}$ is already near-optimal in empirical GFT, but it need not satisfy the positive revenue-margin constraint.  The following posted-spread mechanism supplies the missing revenue slack.
\begin{lemma}\label{lem:spread-main}
There is a universal constant $c_{\rm spr}>0$ such that every product distribution $D=D_B\times D_S$ supported on $[0,h]^2$ with $\FB(D)>0$ admits a posted-spread mechanism $P$ satisfying 
\[ 
\Rev(P;D)\ge c_{\rm spr}\frac{\FB(D)}{1+\log(h/\FB(D))} \ge c_{\rm spr}\frac{\SB(D)}{1+\log(h/\SB(D))},\quad \GFT(P;D)\ge \Rev(P;D)\ge 0. 
\]
In particular, if $\eps\le \SB(D)$, then $\Rev(P;D)\ge c_{\rm spr}\frac{\SB(D)}{\Lambda_\eps}. $
\end{lemma}

Let $P_{\rm emp}:=\Pi_{\widehat D}P$ be the empirical-grid projection of the posted-spread mechanism.  By Lemmas~\ref{lem:spread-main}, \ref{lem:projection}, and~\ref{lem:add-transfer-main}, $P_{\rm emp}$ retains empirical revenue of order $\SB(D)/\Lambda_\eps$.  Mixing $P_{\rm emp}$ into $M^*_{\rm emp}$ with weight $\omega=O(\eps/\SB(D))$ therefore gives
\[
M^\dagger=(1-\omega)M^*_{\rm emp}+\omega P_{\rm emp}
\]
with
\[
\Rev(M^\dagger;\widehat D)\ge \rho_{\rm add},\qquad \GFT(M^\dagger;\widehat D)\ge \SB(D)-O(\eps).
\]
Thus the empirical-margin LP contains a near-optimal feasible mechanism. By empirical optimality, the mechanism returned by Algorithm~\ref{alg:add-upper-main} has empirical GFT at least $\SB(D)-O(\eps)$ and empirical revenue at least $\rho_{\rm add}$. On the event $\Delta_{\rm K}(D,\widehat D)\le\delta_{\rm K}$, Lemma~\ref{lem:add-transfer-main} therefore gives
\[
\GFT(\widehat M;D)\ge \SB(D)-O(\eps)-O(h\delta_{\rm K}),\qquad \Rev(\widehat M;D)\ge \rho_{\rm add}-O(h\delta_{\rm K}).
\]
Choosing $c_{\rm K}$ sufficiently small relative to $c_{\rm add}$ makes the revenue nonnegative and the total GFT loss at most $\eps$.

Finally, by the Dvoretzky--Kiefer--Wolfowitz inequality applied to the two marginals,
\[
\Pr\!\left(\Delta_{\rm K}(D,\widehat D)>\delta_{\rm K}\right)\le 4e^{-2m\delta_{\rm K}^2}.
\]
Combining the above argument with this bound yields the following theorem.

\begin{theorem}\label{thm:add-upper-main}
There are universal constants $C,c_0>0$ such that, for every regular product distribution $D=D_B\times D_S$ supported on $[0,h]^2$ with positive densities and $\SB(D)>0$, and every $\eps\le c_0\SB(D)$, Algorithm~\ref{alg:add-upper-main} with
\[
m\ge C\frac{h^2\Lambda_\eps^2}{\eps^2}\log\frac1\gamma
\]
samples of each marginal returns, with probability at least $1-\gamma$, a BIC, IIR, and ex-ante WBB mechanism $\widehat M$ satisfying
$\GFT(\widehat M;D)\ge \SB(D)-\eps.$
\end{theorem}


\subsection{Lower Bound}\label{ssec:add-lower-main}

This subsection proves the additive lower bound by constructing two regular product distributions on $[0,h]^2$ and showing that the upper bound from \cref{ssec:add-upper-main} is tight up to the logarithmic factors. Let $D_S$ be the uniform distribution on $[0,h]$.  Fix two nonzero nonnegative functions $r_H,r_L\in C_c^2((0,1))$ with equal integrals and
\[
\supp(r_H)\subset(1/2,2/3),
\qquad
\supp(r_L)\subset(0,1/8).
\]
Set $r:=r_H-r_L$.  For $\theta\in\{+1,-1\}$ and sufficiently small $\eta>0$, let
$D_B^\theta$ be the distribution on $[0,h]$ with density
\[
f_B^\theta(b)
=
\frac1h\bigl(1+\theta\eta r(b/h)\bigr),
\qquad b\in[0,h],
\]
and set $D^\theta:=D_B^\theta\times D_S$.
The equality of the integrals ensures that $f_B^\theta$ is a probability density.  Hence $D^+$ and $D^-$ share the seller marginal and differ only by the signed buyer perturbation $\pm\eta r$.  For sufficiently small $\eta$, both instances are regular product distributions; the verification is given in Appendix~\ref{sec:bsu}.  The following lemma bounds their KL divergence.

\begin{lemma}\label{lem:additive-kl}
For the pair above, $\KL(D^+\|D^-)=O(\eta^2) $ and $\ \KL(D^-\|D^+)=O(\eta^2)$.
\end{lemma}

The same pair also separates the WBB constraints: any mechanism that is ex-ante WBB under $D^-$ loses $\Omega(\eta h)$ GFT when evaluated under $D^+$.
\begin{lemma}\label{lem:additive-cross-world-separation}
There are constants $c,\eta_0>0$ such that, for every $\eta\in(0,\eta_0]$, every BIC/IIR mechanism $M$ that is ex-ante WBB under $D^-$ satisfies
\[
\GFT(M;D^+)\le \SB(D^+)-c\eta h.
\]
\end{lemma}

Lemma~\ref{lem:additive-kl} gives the information-distance estimate needed for the two-point argument.  In contrast, Lemma~\ref{lem:additive-cross-world-separation} separates the near-optimal feasible classes of the two instances: a mechanism that succeeds on one instance cannot also satisfy the WBB and near-optimality requirements on the other.  Thus uniform success over the pair forces the learner to distinguish the two buyer marginals, yielding the additive lower bound.

\begin{theorem}\label{thm:add-lower-main}
There are universal constants $c_0,c_1>0$\footnote{All constants in this subsection may depend on the fixed bump functions $r_H,r_L$, which are chosen once and for all, and are independent of the scaling and accuracy parameters.}
such that, for every $h>0$ and $\eps\in(0,c_1h]$, any learner that, uniformly over regular product distributions on $[0,h]^2$, outputs an additive $\eps$-approximation to $\SB(D)$ with probability at least $2/3$ requires
\[
m\ge c_0\frac{h^2}{\eps^2}
\]
samples from each marginal.
\end{theorem}


\section{Multiplicative Approximation with Bounded Support Distributions}\label{sec:main-bounded-mult}

\subsection{Upper Bound}\label{ssec:bounded-mult-upper-main}

This section turns to the bounded-support multiplicative learner. As in \cref{ssec:add-upper-main}, the algorithm is a plug-in empirical optimizer with canonical Myerson payments and a positive empirical revenue margin. The new step is to localize the empirical LP using a pilot estimate of the benchmark scale. This replaces the uniform Kolmogorov estimate with a benchmark-localized bound.

The samples are split into an independent pilot part and a training part.  On the pilot split, pair the buyer and seller samples independently and set
\[ \widehat V:=\frac1{m_0}\sum_{i=1}^{m_0}(b_i-s_i)^+, \qquad \widehat\Lambda:=1+\log\frac h{\widehat V},
\qquad \rho:=c_\rho\frac{\alpha\widehat V}{\widehat\Lambda}, \]
where $c_\rho>0$ is a sufficiently small universal constant, fixed independently of $D,h,\alpha,\gamma$.
In Lemma~\ref{lem:bounded-mult-scale}, we will show that $\widehat V$ is a constant-factor estimate of
$\FB(D)$, and hence of $\SB(D)$ by Lemma~\ref{lem:scale-main}.

\xhdr{Mechanism construction.} 
On the training split, we first construct the empirical product distribution and empirical grids as in \cref{ssec:add-upper-main}, with the same monotone allocation constraints and canonical payments.  Then we use the following LP to construct an allocation rule on the empirical grid, which is converted to a mechanism on the continuous distribution in the same way as \cref{alg:add-upper-main}.
\begin{equation}\label{eq:mult-localized-lp}
\max_{q\in\mathcal Q_{B,S}}\ \widehat{\GFT}(q)
\quad\text{s.t.}\quad
q_{ij}=0\ (\beta_i<\sigma_j),\quad
\widehat{\Rev}(q)\ge 4\rho,\quad
\widehat{\GFT}(q)\le C_0\widehat V .
\end{equation}

Compared with \cref{ssec:add-upper-main}, the new constraints remove negative-surplus trades and localize the empirical mechanism at the scale $\widehat V$. The program remains a polynomial-size finite LP. Since the LP enforces positive revenue, the cap $\widehat{\GFT}(q)\le C_0\widehat V$ also gives $\widehat{\GFT}(q)-\widehat{\Rev}(q)\le C_0\widehat V$. Thus both quantities needed for the localized transfer bound are controlled at the estimated scale $\widehat V=\Theta(\SB(D))$.

\begin{algorithm}[htbp]
\caption{Bounded multiplicative localized learner}\label{alg:mult-upper-main}
\KwIn{$h,\alpha,\gamma$, pilot samples $b^0_1,\ldots,b^0_{m_0}$, $s^0_1,\ldots,s^0_{m_0}$, and training samples $b_1,\ldots,b_m$, $s_1,\ldots,s_m$}
\KwOut{A randomized threshold mechanism on $[0,h]^2$, or no trade}
Compute $\widehat V:=m_0^{-1}\sum_{i=1}^{m_0}(b^0_i-s^0_i)^+$\;
\If{$\widehat V=0$}{
    \Return no trade\;
}
Set $\widehat\Lambda:=1+\log(h/\widehat V)$ and
$\rho:=c_\rho\alpha\widehat V/\widehat\Lambda$\;
Build the training empirical distribution and empirical grids as in \cref{ssec:add-upper-main}\;
Solve the localized empirical LP~\eqref{eq:mult-localized-lp}\;
\If{the LP is infeasible}{
    \Return no trade\;
}
Implement a suitable optimal LP solution using the finite-grid procedure from \cref{ssec:add-upper-main}\;
\Return the resulting randomized threshold mechanism\;
\end{algorithm}

The finite-grid implementation is the same order-polytope construction as in Lemma~\ref{lem:lp}; as shown in Appendix~C, a suitable optimum of the localized LP can be implemented as a mixture of at most three deterministic threshold mechanisms. Hence the output is a randomized threshold mechanism satisfying the empirical revenue margin and the empirical GFT cap in the LP.

\begin{lemma}\label{lem:bounded-mult-scale}
There is a universal constant $C>0$ such that, if the pilot split has
\[ m_0\ge C\frac{h}{\FB(D)}\log\frac1\gamma \]
samples from each marginal, then the pilot estimate $\widehat V$ lies between $\frac12\FB(D)$ and $\frac32\FB(D)$ with probability at least $1-\gamma$. Consequently, $\widehat V=\Theta(\SB(D))$.
\end{lemma}
The following lemma transfers empirical guarantees back to the true distribution for localized mechanisms.  Write $U(M;D'):=\GFT(M;D')-\Rev(M;D').$
For $A>0$, say that $M$ is $A$-localized under $D'$ if
$
\GFT(M;D')\le A,
$ 
and 
$
U(M;D')\le A.
$ Then we have the following lemma.

\begin{lemma}\label{lem:localized-transfer}
With probability at least $1-\gamma$, the following holds uniformly over all $A>0$ and all randomized threshold mechanisms $M$ with allocation $q(b,s)=0$ whenever $b<s$. Denote $\ell:=\log\!\left(\frac{m}{\gamma}\right)$, if $M$ is $A$-localized under either $D$ or $\widehat D$, then 
\[
\bigl|\GFT(M;D)-\GFT(M;\widehat D)\bigr|\le O \left(\sqrt{\frac{hA\ell}{m}}+\frac{h\ell}{m}\right),\quad
\bigl|U(M;D)-U(M;\widehat D)\bigr|\le O\left(\sqrt{\frac{hA\ell}{m}}+\frac{h\ell}{m}\right).
\]
Consequently, $\bigl|\Rev(M;D)-\Rev(M;\widehat D)\bigr|\le  O\left(\sqrt{\frac{hA\ell}{m}}+\frac{h\ell}{m}\right)$.
\end{lemma}

The pilot estimate identifies the benchmark scale, while the localization constraint restricts the empirical program to mechanisms whose GFT and utility are of that scale. Lemma~\ref{lem:localized-transfer} then applies with localization parameter $A=\Theta(\SB(D))$, so the empirical-to-true error is controlled at the benchmark scale rather than the ambient scale $h$. This benchmark-localized comparison yields the improved linear dependence on $h/\SB(D)$.

\begin{theorem}\label{thm:bounded-mult-upper-main}
For every regular product distribution $D$ on $[0,h]^2$ with positive densities, $\SB(D)>0$, and $\alpha\in(0,1/2)$,
Algorithm~\ref{alg:mult-upper-main} with
\[
\Oe\!\left(\frac{h}{\SB(D)}\frac{\Lambda_h(D)^2}{\alpha^2}\log\frac1\gamma\right)
\]
samples from each marginal returns, with probability at least $1-\gamma$, a BIC, IIR, and ex-ante WBB mechanism $\widehat M$ satisfying
$\GFT(\widehat M;D)\ge (1-\alpha)\SB(D).$
\end{theorem}

\subsection{Lower Bound}\label{ssec:bounded-mult-lower-main}
The lower bound shows that the leading factor $h/\SB(D)$ in Theorem~\ref{thm:bounded-mult-upper-main} is unavoidable.
More explicitly, for each rare-block probability $p$ we construct a pair $D_p^+,D_p^-$ with $\SB(D_p^\pm)=\Theta(ph)$.  We first prove a lower bound of $m=\Omega(1/(p\alpha^2))$ for this pair, and the displayed theorem follows by rewriting $p=\Theta(\SB(D_p^\pm)/h)$.
Take the additive hard pair from Section~\ref{ssec:add-lower-main} on the value scale $R=\Theta(h)$, and embed it in a probability $p$ active buyer block.  
Here $R$ is the internal value scale of the additive hard pair; in the formal proof we set $R=h/2$ so that the shifted support $[0,2R]^2$ lies inside $[0,h]^2$.  
Let the scale-$R$ additive pair be $H_R^\theta=X_R^\theta\times Y_R$, let $x_\theta(u)=F^{-1}_{X_R^\theta}(1-u)$ be the buyer upper-tail quantile, and define the upper-tail quantile function of the embedded buyer by
\[
b_\theta(q)=
\begin{cases}
R+x_\theta(q/p), & 0<q\le p,\\[1mm]
pR/q, & p<q\le 1
\end{cases}.
\]
Let the seller's cost be right shifted to $S=R+Y_R$.
The active block is the tail event $q\le p$; it has probability $p$, and conditional on this event the instance is exactly the additive hard pair shifted up by $R$.

Outside the active block, trade has nonpositive surplus and nonpositive virtual surplus, so the WBB separation is entirely carried by the active block. The two embedded instances are identical on this inactive block, and differ only inside the probability-$p$ active block; hence the one-sample KL divergence is the additive-pair KL multiplied by $p$.
Overall, we can show that
\[
\SB(D_p^\pm)=\Theta(pR),\qquad
\KL(D_p^+\|D_p^-), \KL(D_p^-\|D_p^+)=O(p\eta^2).
\]
Choosing $\eta=\Theta(\alpha)$ gives WBB separation $\Omega(\alpha ph)$ and one-sample KL divergence $O(p\alpha^2)$, while $\SB(D^\pm_p)=\Theta(ph)$.  The two-point testing argument therefore yields the following lower bound.

\begin{theorem}\label{thm:bounded-mult-lower-main}
There exists a universal constant $\alpha_0>0$ such that, for every
$\alpha\in(0,\alpha_0)$, any learner that, uniformly over regular product distributions on
$[0,h]^2$, outputs a BIC, IIR, and ex-ante WBB mechanism achieving a $(1-\alpha)$-approximation to $\SB(D)$ with probability at least $2/3$ requires
\[
m=\Omega\!\left(\frac{h}{\SB(D)}\frac1{\alpha^2}\right)
\]
samples from each marginal in the worst case.
\end{theorem}

The upper and lower bounds match up to a factor $\Lambda_h(D)^2$.  The rare-block construction is unnecessary in the additive regime because the target error is absolute: a probability-$p$ block would require $p\eta h=\Theta(\eps)$, hence $\eta=\Theta(\eps/(ph))$, giving one-sample KL divergence
$O(\eps^2/(p h^2))$ and only $\Omega(p h^2/\eps^2)$ samples.


\section{Multiplicative Approximation with Unbounded MHR Distributions}\label{sec:main-mhr}

We extend the bounded multiplicative learner to unbounded instances.  The difficulty is that the bounded-support scale $h$ of Section~\ref{sec:main-bounded-mult} is absent.  Under MHR marginals and seller regularity, the learner estimates a finite cap $\widehat w$, forms a truncated instance, runs the bounded learner on that instance, and extends the learned mechanism back to the original type space.

\begin{assumption}[MHR marginals and seller regularity]\label{ass:mhr-main}
The buyer and seller marginals are supported on intervals in $[0,\infty)$, have finite means and strictly positive densities on their support interiors, and are MHR.  In addition, the seller virtual cost
$\phi_S(s):=s+\frac{F_S(s)}{f_S(s)}$
is nondecreasing on the seller support interior.
\end{assumption}

The main proof idea is to turn the unbounded instance into a bounded one without creating spurious trades in the tail. A direct cap works naturally for buyer values, but it is not safe for seller costs: replacing a very high cost by the cap could make an infeasible trade appear profitable. We therefore use a cap/sentinel truncation.
For $w>0$, define the truncated distribution $D_w$ by
\[
B_w:=\min\{B,w\},\qquad
S_w:=
\begin{cases}
S,& S\le w,\\
\top_S,& S>w,
\end{cases}
\]
where $\top_S$ is a no-trade seller sentinel type. The sentinel type prevents high seller costs from being capped downward and creating artificial surplus.  Mechanisms on $D_w$ are required to allocate zero on the seller-sentinel type.  Such a mechanism $N$ is run on the original instance through the truncation map; denote the resulting mechanism by $\Ext_w(N)$.

Write
\[
\mu(D):=\E[B]+\E[S],
\chi_\mu(D):=\frac{\mu(D)}{\SB(D)},
L_{\mu,\alpha}(D):=1+\log\frac{\chi_\mu(D)}{\alpha}.
\]
The learner uses an independent sample split to estimate the mean scale and to find a posted-spread mechanism with a lower confidence bound on its revenue.  More precisely, Appendix~\ref{sec:mhr} shows that, with high probability, it obtains $\widehat\mu=\Theta(\mu(D))$ and $\widehat R=\Omega(\SB(D)/L_{\mu,\alpha}(D))$ with $\widehat R\le \Rev(P;D)$. It then sets
\[
\widehat w:=C_{\rm cap}\widehat\mu\left(1+\log\frac{2\widehat\mu}{\alpha\widehat R}\right).
\]
For a sufficiently large universal constant $C_{\rm cap}$, these estimates together with the MHR tail bound give
\[
\SB(D_{\widehat w})\ge \left(1-\frac{\alpha}{2}\right)\SB(D),
\qquad
\frac{\widehat w}{\SB(D_{\widehat w})}
=O\!\left(\chi_\mu(D)L_{\mu,\alpha}(D)\right).
\]
Thus the bounded-support multiplicative learner of Section~\ref{sec:main-bounded-mult} can be applied to $D_{\widehat w}$ with accuracy $\alpha/2$, after which the learned mechanism is extended back to the original type space through $\Ext_{\widehat w}$. The extension preserves revenue and does not decrease GFT. This gives the following bound.

\begin{theorem}\label{thm:mhr-upper-main}
Under Assumption~\ref{ass:mhr-main}, for every product distribution $D$ with $\SB(D)>0$, every $\alpha\in(0,1/2)$, and every $\gamma\in(0,1)$, the estimated-cap learner described above and formalized in Algorithm~\ref{alg:mhr-main} with
\[
\Oe\!\left(
\frac{\chi_\mu(D)L_{\mu,\alpha}(D)^3}{\alpha^2}
\log\frac1\gamma
\right)
\]
samples from each marginal returns, with probability at least $1-\gamma$, a BIC, IIR, and ex-ante WBB mechanism $\widehat M$ satisfying $\GFT(\widehat M;D)\ge (1-\alpha)\SB(D)$.
\end{theorem}

The lower bound shows that this benchmark-sensitive dependence is not an artifact of the cap. Write $L_\chi(D):=1+\log\chi_\mu(D)$. Our construction hides the bounded-support hard instance inside a rare upper-tail region of an MHR distribution. Outside that region, the two candidate worlds are essentially indistinguishable and do not provide useful budget-feasible trade. Inside the rare region, the same WBB separation as in the bounded lower bound reappears, but the learner sees that region only with its small tail probability. Thus the benchmark itself is small, yet achieving a relative approximation still requires discovering the rare informative part of the market. This yields the following lower bound.


\begin{theorem}\label{thm:mhr-lower-main}
For arbitrarily large benchmark ratios $\chi$ and all sufficiently small $\alpha>0$, there are product distributions satisfying Assumption~\ref{ass:mhr-main} with
$\chi_\mu(D)=\Theta(\chi)$ such that any learner that, with probability at least $2/3$, outputs a BIC, IIR, and ex-ante WBB mechanism achieving a $(1-\alpha)$-approximation to $\SB(D)$ requires
\[
m=\Omega\!\left(
\frac{\chi_\mu(D)}{L_\chi(D)\alpha^2}
\right)
\]
samples from each marginal in the worst case.
\end{theorem}

Thus the MHR upper and lower bounds match up to polylogarithmic factors $L_{\mu,\alpha}(D)^3L_\chi(D)$.

\section*{Acknowledgments}
We used generative AI tools as an interactive aid in developing parts of the proofs, as well as in improving the presentation of the paper. All mathematical arguments and claims were independently reviewed and verified by the authors, who take full responsibility for the content of the paper.

\bibliographystyle{plainnat}
\bibliography{mybib}

@article{AzarFiatFusco2024,
  title   = {An {$\alpha$}-regret analysis of adversarial bilateral trade},
  author  = {Azar, Yossi and Fiat, Amos and Fusco, Federico},
  journal = {Artificial Intelligence},
  volume  = {337},
  pages   = {104231},
  year    = {2024},
  doi     = {10.1016/j.artint.2024.104231}
}

@inproceedings{BabaioffFreyNisan2024,
  title     = {Learning to Maximize Gains From Trade in Small Markets},
  author    = {Babaioff, Moshe and Frey, Amitai and Nisan, Noam},
  booktitle = {Proceedings of the 25th ACM Conference on Economics and Computation (EC)},
  pages     = {195},
  year      = {2024},
  doi       = {10.1145/3670865.3673463}
}

@inproceedings{BernasconiCastiglioniCelliFusco2024,
  title     = {No-regret learning in bilateral trade via global budget balance},
  author    = {Bernasconi, Martino and Castiglioni, Matteo and Celli, Andrea and Fusco, Federico},
  booktitle = {Proceedings of the 56th Annual ACM Symposium on Theory of Computing},
  pages     = {247--258},
  year      = {2024},
  doi       = {10.1145/3618260.3649653}
}

@article{BMP1963,
  title     = {Properties of probability distributions with monotone hazard rate},
  author    = {Barlow, Richard E. and Marshall, Albert W. and Proschan, Frank},
  journal   = {The Annals of Mathematical Statistics},
  volume    = {34},
  number    = {2},
  pages     = {375--389},
  year      = {1963},
  doi       = {10.1214/aoms/1177704147}
}

@inproceedings{BlumrosenMizrahi2016,
  title     = {Approximating Gains-from-Trade in Bilateral Trading},
  author    = {Blumrosen, Liad and Mizrahi, Yehonatan},
  booktitle = {Web and Internet Economics},
  series    = {Lecture Notes in Computer Science},
  volume    = {10123},
  pages     = {400--413},
  publisher = {Springer},
  year      = {2016},
  doi       = {10.1007/978-3-662-54110-4_28}
}

@article{BlumSandholmZinkevich2006,
  title     = {Online algorithms for market clearing},
  author    = {Blum, Avrim and Sandholm, Tuomas and Zinkevich, Martin},
  journal   = {Journal of the ACM},
  volume    = {53},
  number    = {5},
  pages     = {845--879},
  year      = {2006},
  doi       = {10.1145/1183907.1183913}
}

@inproceedings{BCWZ2017,
  title     = {Approximating gains from trade in two-sided markets via simple mechanisms},
  author    = {Brustle, Johannes and Cai, Yang and Wu, Fa and Zhao, Mingfei},
  booktitle = {Proceedings of the 2017 ACM Conference on Economics and Computation},
  pages     = {589--590},
  year      = {2017},
  doi       = {10.1145/3033274.3085148},
  note      = {Full version: arXiv:1706.04637}
}

@inproceedings{CaiWu2023,
  title     = {On the Optimal Fixed-Price Mechanism in Bilateral Trade},
  author    = {Cai, Yang and Wu, Jinzhao},
  booktitle = {Proceedings of the 55th Annual ACM Symposium on Theory of Computing},
  pages     = {737--750},
  year      = {2023},
  doi       = {10.1145/3564246.3585171}
}

@inproceedings{LiuRenWang2023,
  title     = {Improved Approximation Ratios of Fixed-Price Mechanisms in Bilateral Trades},
  author    = {Liu, Zhengyang and Ren, Zeyu and Wang, Zihe},
  booktitle = {Proceedings of the 55th Annual ACM Symposium on Theory of Computing},
  pages     = {751--760},
  year      = {2023},
  doi       = {10.1145/3564246.3585160}
}

@inproceedings{CastiglioniLunghiMarchesi2026,
  title     = {The Sample Complexity of Uniform Approximation for Multi-Dimensional {CDFs} and Fixed-Price Mechanisms},
  author    = {Castiglioni, Matteo and Lunghi, Anna and Marchesi, Alberto},
  booktitle = {Proceedings of the 58th Annual ACM Symposium on Theory of Computing},
  pages     = {1353--1364},
  year      = {2026},
  doi       = {10.1145/3798129.3800845}
}

@article{CesaBianchiEtAl2024a,
  title     = {Bilateral trade: A regret minimization perspective},
  author    = {Cesa-Bianchi, Nicol{\`o} and Cesari, Tommaso and Colomboni, Roberto and Fusco, Federico and Leonardi, Stefano},
  journal   = {Mathematics of Operations Research},
  volume    = {49},
  number    = {1},
  pages     = {171--203},
  year      = {2024},
  doi       = {10.1287/moor.2023.1351}
}

@article{CesaBianchiEtAl2024b,
  title   = {Regret analysis of bilateral trade with a smoothed adversary},
  author  = {Cesa-Bianchi, Nicol{\`o} and Cesari, Tommaso and Colomboni, Roberto and Fusco, Federico and Leonardi, Stefano},
  journal = {Journal of Machine Learning Research},
  volume  = {25},
  number  = {234},
  pages   = {1--36},
  year    = {2024},
  url     = {https://jmlr.org/papers/v25/23-1627.html}
}

@article{ChenJinLuZhang2025,
  title   = {Tight Regret Bounds for Fixed-Price Bilateral Trade},
  author  = {Chen, Houshuang and Jin, Yaonan and Lu, Pinyan and Zhang, Chihao},
  journal = {arXiv preprint arXiv:2504.04349},
  year    = {2025},
  note    = {To appear in ICALP 2026}
}

@inproceedings{DHP2016,
  title     = {The Sample Complexity of Auctions with Side Information},
  author    = {Devanur, Nikhil R. and Huang, Zhiyi and Psomas, Christos-Alexandros},
  booktitle = {Proceedings of the 48th Annual ACM Symposium on Theory of Computing},
  pages     = {426--439},
  year      = {2016},
  doi       = {10.1145/2897518.2897553},
  note      = {Full version: arXiv:1511.02296}
}

@inproceedings{DiGregorioDuttingFuscoSchwiegelshohn2025,
  title     = {Nearly Tight Regret Bounds for Profit Maximization in Bilateral Trade},
  author    = {Di Gregorio, Simone and D{\"u}tting, Paul and Fusco, Federico and Schwiegelshohn, Chris},
  booktitle = {Proceedings of the 66th IEEE Annual Symposium on Foundations of Computer Science (FOCS)},
  pages     = {1570--1594},
  year      = {2025},
  doi       = {10.1109/FOCS63196.2025.00083}
}

@article{DiGregorioFuscoLeonardiSchwiegelshohn2026,
  title   = {Private Learning in Bilateral Trade},
  author  = {Di Gregorio, Simone and Fusco, Federico and Leonardi, Stefano and Schwiegelshohn, Chris},
  journal = {arXiv preprint arXiv:2606.02050},
  year    = {2026}
}

@inproceedings{DengMaoSivanWangWu2025,
  title     = {Approximately Efficient Bilateral Trade with Samples},
  author    = {Deng, Yuan and Mao, Jieming and Sivan, Balasubramanian and Wang, Kangning and Wu, Jinzhao},
  booktitle = {Proceedings of the 26th ACM Conference on Economics and Computation},
  pages     = {206--223},
  year      = {2025},
  doi       = {10.1145/3736252.3742519}
}

@inproceedings{DMSW2022,
  title     = {Approximately efficient bilateral trade},
  author    = {Deng, Yuan and Mao, Jieming and Sivan, Balasubramanian and Wang, Kangning},
  booktitle = {Proceedings of the 54th Annual ACM Symposium on Theory of Computing},
  pages     = {718--721},
  year      = {2022},
  doi       = {10.1145/3519935.3520054}
}

@inproceedings{DobzinskiShaulker2024,
  title     = {Bilateral Trade with Correlated Values},
  author    = {Dobzinski, Shahar and Shaulker, Ariel},
  booktitle = {Proceedings of the 56th Annual ACM Symposium on Theory of Computing},
  pages     = {237--246},
  year      = {2024},
  doi       = {10.1145/3618260.3649659}
}

@inproceedings{DobzinskiEdenGoldnerShaulkerTsilivis2025,
  title     = {Bilateral Trade with Interdependent Values: Information vs. Approximation},
  author    = {Dobzinski, Shahar and Eden, Alon and Goldner, Kira and Shaulker, Ariel and Tsilivis, Thodoris},
  booktitle = {Proceedings of the 26th ACM Conference on Economics and Computation},
  pages     = {641--665},
  year      = {2025},
  doi       = {10.1145/3736252.3742607}
}

@inproceedings{DuttingFuscoLazosLeonardiReiffenhauser2021,
  title     = {Efficient Two-Sided Markets with Limited Information},
  author    = {D{\"u}tting, Paul and Fusco, Federico and Lazos, Philip and Leonardi, Stefano and Reiffenh{\"a}user, Rebecca},
  booktitle = {Proceedings of the 53rd Annual ACM SIGACT Symposium on Theory of Computing},
  pages     = {1452--1465},
  year      = {2021},
  doi       = {10.1145/3406325.3451076}
}

@inproceedings{Fei2022,
  title     = {Improved approximation to first-best gains-from-trade},
  author    = {Fei, Yumou},
  booktitle = {Web and Internet Economics},
  series    = {Lecture Notes in Computer Science},
  volume    = {13778},
  pages     = {204--218},
  publisher = {Springer},
  year      = {2022},
  doi       = {10.1007/978-3-031-22832-2_12}
}

@inproceedings{FengMaPengWan2026,
  author    = {Yiding Feng and Mengfan Ma and Bo Peng and Zongqi Wan},
  title     = {Searching for Optimal Prices in Two-Sided Markets},
  booktitle = {Proceedings of the 27th ACM Conference on Economics and Computation (EC)},
  year      = {2026},
}

@inproceedings{GaucherBernasconiCastiglioniCelliPerchet2025,
  title     = {Feature-Based Online Bilateral Trade},
  author    = {Gaucher, Solenne and Bernasconi, Martino and Castiglioni, Matteo and Celli, Andrea and Perchet, Vianney},
  booktitle = {The Thirteenth International Conference on Learning Representations},
  year      = {2025}
}

@inproceedings{GuoHuangZhang2019,
  title     = {Settling the sample complexity of single-parameter revenue maximization},
  author    = {Guo, Chenghao and Huang, Zhiyi and Zhang, Xinzhi},
  booktitle = {Proceedings of the 51st Annual ACM Symposium on Theory of Computing},
  pages     = {662--673},
  year      = {2019},
  doi       = {10.1145/3313276.3316325}
}

@article{HartlineWang2025,
  title   = {A Geometric Analysis of Gains from Trade},
  author  = {Hartline, Jason and Wang, Kangning},
  journal = {arXiv preprint arXiv:2508.06469},
  year    = {2025}
}

@article{Jin2026,
  title   = {Tight Regret Bounds for Bilateral Trade under Semi Feedback},
  author  = {Jin, Yaonan},
  journal = {arXiv preprint arXiv:2601.16412},
  year    = {2026}
}

@inproceedings{KangPerniceVondrak2022,
  title        = {Fixed-Price Approximations in Bilateral Trade},
  author       = {Kang, Zi Yang and Pernice, Francisco and Vondr{\'a}k, Jan},
  booktitle    = {Proceedings of the 2022 Annual ACM-SIAM Symposium on Discrete Algorithms (SODA)},
  pages        = {2964--2985},
  year         = {2022},
  organization = {SIAM},
  doi          = {10.1137/1.9781611977073.115}
}

@article{KunimotoZhang2026,
  title   = {Efficient Bilateral Trade with Interdependent Values: The Use of Two-Stage Mechanisms},
  author  = {Kunimoto, Takashi and Zhang, Cuiling},
  journal = {Journal of Mathematical Economics},
  volume  = {122},
  pages   = {103191},
  year    = {2026},
  doi     = {10.1016/j.jmateco.2025.103191}
}

@article{LQRW2026,
  title   = {Second-Best Bilateral Trade is {$1/2$} Efficient},
  author  = {Liu, Zhengyang and Qin, Ying and Ren, Zeyu and Wang, Zihe},
  journal = {arXiv preprint arXiv:2606.03849},
  year    = {2026}
}

@inproceedings{LunghiCastiglioniMarchesi2026,
  title        = {Better Regret Rates in Bilateral Trade via Sublinear Budget Violation},
  author       = {Lunghi, Anna and Castiglioni, Matteo and Marchesi, Alberto},
  booktitle    = {Proceedings of the 2026 Annual ACM-SIAM Symposium on Discrete Algorithms (SODA)},
  pages        = {6494--6536},
  year         = {2026},
  organization = {SIAM},
  doi          = {10.1137/1.9781611978971.233}
}

@inproceedings{LunghiPiccinatoCastiglioniMarchesi2026,
  title     = {A Stronger Benchmark for Online Bilateral Trade: From Fixed Prices to Distributions},
  author    = {Lunghi, Anna and Piccinato, Mattia and Castiglioni, Matteo and Marchesi, Alberto},
  booktitle = {Proceedings of the 43rd International Conference on Machine Learning (ICML)},
  year      = {2026}
}

@article{DvoretzkyKieferWolfowitz1956,
  title   = {Asymptotic Minimax Character of the Sample Distribution Function and of the Classical Multinomial Estimator},
  author  = {Dvoretzky, Aryeh and Kiefer, Jack and Wolfowitz, Jacob},
  journal = {The Annals of Mathematical Statistics},
  volume  = {27},
  number  = {3},
  pages   = {642--669},
  year    = {1956},
  doi     = {10.1214/aoms/1177728174}
}

@article{Massart1990,
  title   = {The tight constant in the {Dvoretzky--Kiefer--Wolfowitz} inequality},
  author  = {Massart, Pascal},
  journal = {The Annals of Probability},
  volume  = {18},
  number  = {3},
  pages   = {1269--1283},
  year    = {1990},
  doi     = {10.1214/aop/1176990746}
}

@article{McAfee2008,
  title   = {The gains from trade under fixed price mechanisms},
  author  = {McAfee, R. Preston},
  journal = {Applied Economics Research Bulletin},
  volume  = {1},
  number  = {1},
  pages   = {1--10},
  year    = {2008}
}

@article{Myerson1981,
  title   = {Optimal auction design},
  author  = {Myerson, Roger B.},
  journal = {Mathematics of Operations Research},
  volume  = {6},
  number  = {1},
  pages   = {58--73},
  year    = {1981},
  doi     = {10.1287/moor.6.1.58}
}

@article{MyersonSatterthwaite1983,
  title     = {Efficient mechanisms for bilateral trading},
  author    = {Myerson, Roger B. and Satterthwaite, Mark A.},
  journal   = {Journal of Economic Theory},
  volume    = {29},
  number    = {2},
  pages     = {265--281},
  year      = {1983},
  doi       = {10.1016/0022-0531(83)90048-0}
}

@article{Panchenko2003,
  title   = {Symmetrization approach to concentration inequalities for empirical processes},
  author  = {Panchenko, Dmitriy},
  journal = {The Annals of Probability},
  volume  = {31},
  number  = {4},
  pages   = {2068--2081},
  year    = {2003},
  doi     = {10.1214/aop/1068646378}
}

@article{Stanley1986,
  title   = {Two poset polytopes},
  author  = {Stanley, Richard P.},
  journal = {Discrete \& Computational Geometry},
  volume  = {1},
  number  = {1},
  pages   = {9--23},
  year    = {1986},
  doi     = {10.1007/BF02187680}
}

@book{Tsybakov2009,
  title     = {Introduction to Nonparametric Estimation},
  author    = {Tsybakov, Alexandre B.},
  series    = {Springer Series in Statistics},
  publisher = {Springer},
  address   = {New York},
  year      = {2009},
  doi       = {10.1007/b13794},
  isbn      = {9780387790510}
}

\clearpage
\appendix

\section{Threshold Mechanisms on Grids}\label{sec:prelim}

\subsection{Second-Best Threshold Structure}
\begin{proof}[Proof of Theorem~\ref{fact:structure-main}]
The second-best characterization of \citet{MyersonSatterthwaite1983,BCWZ2017} gives an optimal allocation that, for some $\lambda^*\ge0$, maximizes pointwise
\[
q\Big((b-s)+\lambda^*(\phi_B(b)-\phi_S(s))\Big),
\qquad q\in[0,1].
\]
Hence, outside a null tie set,
\[
q^*(b,s)=\one\{G(b)\ge H(s)\},
\qquad
G(b):=b+\lambda^*\phi_B(b),
\quad
H(s):=s+\lambda^*\phi_S(s).
\]
Regularity makes $G$ and $H$ nondecreasing, so $q^*$ is nondecreasing in $b$ and nonincreasing in $s$, hence is a canonical threshold allocation.  The tie set $\{G(b)=H(s)\}$ is a monotone graph and has $D$-measure zero because both marginals have densities.  Myerson's payment identity~\citep{Myerson1981} implements $q^*$ as a BIC and IIR mechanism, and the second-best characterization gives ex-ante WBB and value $\SB(D)$.
\end{proof}

\subsection{Finite-grid implementation}
\begin{lemma}[\citealp{Myerson1981}]\label{lem:fin-impl-main}
On finite buyer and seller grids
\[
B=\{\beta_0<\cdots<\beta_K\},
\qquad
S=\{\sigma_0<\cdots<\sigma_L\},
\]
every allocation $q$ that is nondecreasing in the buyer report and nonincreasing in the seller report is implemented by
\[
x_B(\beta_i,\sigma_j)
=\beta_iq(\beta_i,\sigma_j)-\sum_{t=0}^{i-1}(\beta_{t+1}-\beta_t)q(\beta_t,\sigma_j)
\]
and
\[
x_S(\beta_i,\sigma_j)
=\sigma_jq(\beta_i,\sigma_j)+\sum_{t=j+1}^{L}(\sigma_t-\sigma_{t-1})q(\beta_i,\sigma_t).
\]
The lowest buyer type and the highest seller type obtain zero utility.  Thus the mechanism is DSIC and ex-post IR on the grid, hence BIC and IIR.  Integral monotone allocations are deterministic threshold rules; fractional allocations used below are implemented as report-independent mixtures of integral rules.
\end{lemma}

\begin{lemma}\label{lem:grid-step-extension}
Let $q:B\times S\to\{0,1\}$ be monotone on the grids above.  Define
\[
\pi_B(b):=\max\{\beta_i:\beta_i\le b\},
\qquad
\pi_S(s):=\min\{\sigma_j:\sigma_j\ge s\},
\]
and, on $[\beta_0,\beta_K]\times[\sigma_0,\sigma_L]$,
\[
(\operatorname{Step}_{B,S}q)(b,s):=q(\pi_B(b),\pi_S(s)).
\]
Then $\operatorname{Step}_{B,S}q$ is monotone, agrees with $q$ on the grid, admits threshold representations
\[
\one\{b\ge\tau_B(s)\}=\operatorname{Step}_{B,S}q(b,s)=\one\{s\le\tau_S(b)\},
\]
and its threshold payments agree on grid profiles with the envelope payments in Lemma~\ref{lem:fin-impl-main}.  Consequently, componentwise stepwise extension preserves DSIC and ex-post IR for report-independent mixtures and leaves all expectations unchanged under distributions supported on the grid.
\end{lemma}

\begin{proof}
Monotonicity follows because both coordinate projections are nondecreasing and the seller projection rounds upward.  For fixed $s$, traded buyer reports form an upper interval; for fixed $b$, traded seller reports form a lower interval.  Taking inclusive endpoints gives the displayed threshold rules. At grid profiles these critical thresholds are exactly the adjacent grid thresholds in the discrete envelope formulas, so payments agree with those in Lemma~\ref{lem:fin-impl-main}.  The statements for mixtures and grid-supported expectations follow componentwise.
\end{proof}

\xhdr{Threshold convention.}
Deterministic threshold allocations use inclusive events $\{b\ge\tau_B(s)\}$ and $\{s\le\tau_S(b)\}$. Thus buyer threshold probabilities use the closed upper tail $\Pr[B\ge t]=1-F_B(t^-)$, seller threshold probabilities use $F_S(t)$, and boundary thresholds follow the conventions $\Pr[B\ge+\infty]=0$, $\Pr[B\ge-\infty]=1$, $F_S(-\infty)=0$, and $F_S(+\infty)=1$.  At jump points, we use left-continuous buyer thresholds and right-continuous seller thresholds.

\section{Proofs for Section~\ref{sec:main-additive}}\label{sec:bsu}

\subsection{Proof of Lemma~\ref{lem:add-transfer-main}}

\begin{proof}
It suffices to prove the statement for deterministic canonical threshold mechanisms; the randomized case follows by linearity.  We use the following one-dimensional estimate.  If two probability measures on $[0,h]$ have cdfs within distance $\delta$, then their integrals against any bounded-variation function $\psi$ differ by at most $\operatorname{Var}(\psi)\delta$, where representatives are chosen consistently with the threshold convention.  Also, if $f$ is monotone and $[0,h]$-valued and $g$ is monotone and $[0,1]$-valued, then
\[
\operatorname{Var}(fg)
\le \|f\|_\infty\operatorname{Var}(g)+\|g\|_\infty\operatorname{Var}(f)
\le 2h .
\]

Let $M$ be a deterministic canonical threshold mechanism with threshold maps $\tau_B,\tau_S$. Write
\[
\overline F_B(t):=\Pr[B\ge t].
\]
Then
\[
\GFT(M;D)
=
\int_0^h bF_S(\tau_S(b))\,dF_B(b)
-
\int_0^h s\overline F_B(\tau_B(s))\,dF_S(s).
\]
Consider the first term. Replacing $F_S$ by $\widehat F_S$ inside the integrand changes the value by at most $h\delta$.  After this replacement, the integrand $b\widehat F_S(\tau_S(b))$ is the product of a monotone $[0,h]$-valued function and a monotone $[0,1]$-valued function, and hence has variation at most $2h$.  Replacing the outer marginal measure $dF_B$ by $d\widehat F_B$ therefore changes the value by at most $2h\delta$. Hence this term changes by at most $3h\delta$.  Applying the same argument to
$s\overline F_B(\tau_B(s))$ gives the same bound for the second term. Thus
\[
|\GFT(M;D)-\GFT(M;\widehat D)|\le 6h\delta .
\]

For revenue, the threshold-payment formula gives
\[
\Rev(M;D)
=
\int_0^h p_B(s)\overline F_B(\tau_B(s))\,dF_S(s)
-
\int_0^h p_S(b)F_S(\tau_S(b))\,dF_B(b),
\]
where $p_B(s)=\tau_B(s)$ when the threshold is finite, and $p_B(s)=h$ under the no-trade boundary convention.  This convention only keeps the integrand bounded; the associated trade probability is zero except possibly on a boundary of zero true mass.  Similarly, $p_S(b)=\tau_S(b)$ when the threshold is finite and $p_S(b)=0$ under the no-trade boundary convention. Each integrand is again the product of a monotone $[0,h]$-valued function and a monotone $[0,1]$-valued function.  The same argument gives
\[
|\Rev(M;D)-\Rev(M;\widehat D)|\le 6h\delta .
\]
Combining the two estimates proves
\[
|\GFT(M;D)-\GFT(M;\widehat D)|
+
|\Rev(M;D)-\Rev(M;\widehat D)|
\le 12h\delta .
\]
This is the claimed $O(h\delta)$ bound.
\end{proof}

\subsection{Proof of Lemma~\ref{lem:projection}}
\begin{proof}
Let $M$ be a deterministic canonical threshold mechanism with threshold maps $\tau_B,\tau_S$, and let $M_{\rm emp}:=\Pi_{\widehat D}M$.

By definition, the finite-grid allocation of $M_{\rm emp}$ agrees with the allocation of $M$ on every empirical profile $(\beta_i,\sigma_j)\in B_{\rm emp}\times S_{\rm emp}$.  Since $\widehat D$ is supported on these profiles, the empirical GFT is unchanged:
\[
\GFT(M_{\rm emp};\widehat D)=\GFT(M;\widehat D).
\]

It remains to compare empirical revenue.  The finite-grid allocation rounds each finite buyer threshold up to the smallest traded buyer grid point, hence to a value at least $\tau_B(\sigma_j)$, and rounds each finite seller threshold down to the largest traded seller grid point, hence to a value at most $\tau_S(\beta_i)$.  Follow the boundary threshold convention.  Therefore, on every traded empirical profile, the buyer payment weakly increases and the seller payment weakly decreases, giving
\[
\Rev(M_{\rm emp};\widehat D)\ge \Rev(M;\widehat D).
\]
The randomized case follows componentwise.
\end{proof}

\subsection{Proof of Lemma~\ref{lem:spread-main}}

\begin{proof}
Let
\[
R:=\sup_P \Rev(P;D),\qquad W:=\FB(D),
\]
where the supremum is over posted-spread mechanisms.  Since $D$ is supported on $[0,h]^2$, we have $0<W\le h$.  If $R=0$, then
\[
F_S(s)\sup_{y\ge s}(y-s)\overline F_B(y)=0
\]
for every $s$, and hence $A_S(D)=0$.  The symmetric argument gives $A_B(D)=0$, contradicting the delegated-pricing ratio below because $W>0$.  Thus $R>0$.  Since posted-spread revenue is pointwise bounded by realized GFT, $R\le W$.

For every seller type $s$,
\[
R\ge F_S(s)\sup_{y\ge s}(y-s)\overline F_B(y).
\]
Let $Q_S(a):=\inf\{s:F_S(s)\ge a\}$ and
\[
g(a):=\sup_{y\ge Q_S(a)}(y-Q_S(a))\overline F_B(y).
\]
Then $g(a)\le R/a$ and $g(a)\le h$. Hence
\[
A_S(D):=\int_0^1 g(a)\,da
\le \int_0^1 \min\{h,R/a\}\,da
\le R\left(1+\log\frac hR\right).
\]
The symmetric buyer-side argument gives
\[
A_B(D)\le R\left(1+\log\frac hR\right).
\]
By the delegated-pricing ratio of \citet[Theorem~1.1]{Fei2022}, 
\[
W\le 3.15\max\{A_S(D),A_B(D)\}
\le 3.15R\left(1+\log\frac hR\right).
\]
Set $z:=R/W$ and $\Lambda_0:=1+\log(h/W)$.  Then
\[
1\le 3.15z\left(\Lambda_0+\log\frac1z\right).
\]
Since $u\mapsto u(\Lambda_0+\log(1/u))$ is nondecreasing on $(0,1]$, this implies
$z\ge c/\Lambda_0$ for a universal constant $c>0$. Therefore
\[
R\ge c\frac{\FB(D)}{1+\log(h/\FB(D))}
\ge c\frac{\SB(D)}{1+\log(h/\SB(D))},
\]
using $\SB(D)\le \FB(D)\le h$.

Taking a posted-spread mechanism with revenue at least $R/2$ gives the stated revenue bound. Finally, a posted-spread mechanism trades only when the realized GFT dominates the price spread, so $\GFT(P;D)\ge \Rev(P;D)\ge 0$.
\end{proof}

\subsection{Proof of Theorem~\ref{thm:add-upper-main}}

\begin{lemma}[\citealp{DvoretzkyKieferWolfowitz1956,Massart1990}] \label{lem:marginal-kolmogorov}
Let $\widehat D=\widehat D_B\times \widehat D_S$ be formed from $m$ independent samples from each marginal of
$D=D_B\times D_S$.  Then, for every $\delta>0$,
\[
\Pr\!\left(\Delta_{\rm K}(\widehat D,D)>\delta\right)
\le 4e^{-2m\delta^2}.
\]
\end{lemma}

\begin{lemma}\label{lem:lp}
The empirical-margin LP has an optimum implementable as a mixture of at most two deterministic threshold mechanisms with the same empirical GFT and revenue.
\end{lemma}

\begin{proof}
The monotone allocation constraints define the order polytope of the finite grid order. Its vertices are precisely the monotone $\{0,1\}$ allocations \citep[Corollary~1.3]{Stanley1986}.  Intersect this polytope with the revenue halfspace and take an optimal vertex $q^\star$ of the resulting LP.  Let $G$ be the minimal face of the order polytope containing $q^\star$.  If the active revenue constraint has rank smaller than $\dim(G)$ on the tangent space of $G$, then $q^\star$ can be perturbed in both directions inside $G$ while preserving all active constraints, contradicting vertexhood.  Hence $\dim(G)\le1$.  Therefore $G$ is either a vertex or an edge of the order polytope, and $q^\star$ is either a vertex or a convex combination of the two endpoints of that edge.  By Lemmas~\ref{lem:fin-impl-main} and~\ref{lem:grid-step-extension}, the relevant vertices can be implemented as deterministic threshold mechanisms, and empirical GFT and revenue are preserved by linearity.
\end{proof}

\begin{lemma}\label{lem:mixing}
Fix a distribution and mechanisms $M_0,P$ on the same type space with
\[
\GFT(M_0)\ge G,\qquad \Rev(M_0)\ge -d,\qquad \GFT(P)\ge 0,\qquad \Rev(P)\ge R>0 .
\]
If
\[
\omega:=\frac{d+\rho_{\rm tar}}{R}\le 1,
\]
then $\widetilde M=(1-\omega)M_0+\omega P$ satisfies
\[
\Rev(\widetilde M)\ge \rho_{\rm tar},
\qquad
\GFT(\widetilde M)\ge (1-\omega)G .
\]
\end{lemma}

\begin{proof}
Report-independent mixing makes $\Rev$ and $\GFT$ affine in the mixture weight, so
\[
\Rev(\widetilde M)\ge -(1-\omega)d+\omega R\ge -d+\omega R=\rho_{\rm tar},
\qquad
\GFT(\widetilde M)\ge (1-\omega)G .
\]
\end{proof}

\begin{lemma}\label{lem:empirical-mixed-main}
If $\Delta_{\rm K}(\widehat D,D)\le \delta_{\rm K}$, then there is an empirical-grid mechanism $M^\dagger$ such that
\[
\Rev(M^\dagger;\widehat D)\ge \rho_{\rm add},
\qquad
\GFT(M^\dagger;\widehat D)\ge \SB(D)-O(\eps).
\]
\end{lemma}

\begin{proof}
Let $M^\ast$ be a canonical optimum for $D$, and set
$M^\ast_{\rm emp}:=\Pi_{\widehat D}M^\ast$.  Lemmas~\ref{lem:add-transfer-main}
and~\ref{lem:projection} give
\[
\GFT(M^\ast_{\rm emp};\widehat D)\ge \SB(D)-O(h\delta_{\rm K}),
\qquad
\Rev(M^\ast_{\rm emp};\widehat D)\ge -O(h\delta_{\rm K}).
\]
Let $P$ be the posted-spread mechanism from Lemma~\ref{lem:spread-main}, and set
$P_{\rm emp}:=\Pi_{\widehat D}P$.  Lemmas~\ref{lem:add-transfer-main},
\ref{lem:projection}, and~\ref{lem:spread-main} imply
\[
\Rev(P_{\rm emp};\widehat D)
\ge
c_{\rm spr}\frac{\SB(D)}{\Lambda_\eps}-O(h\delta_{\rm K})
=
\Omega\!\left(\frac{\SB(D)}{\Lambda_\eps}\right),
\]
where the last step uses the choice of $\delta_{\rm K}$ and $\eps\le c_0\SB(D)$. Applying Lemma~\ref{lem:mixing} under $\widehat D$ to $M^\ast_{\rm emp}$ and $P_{\rm emp}$, with target margin $\rho_{\rm add}$, gives a mixture weight
\[
\omega
=
O\!\left(\frac{(\rho_{\rm add}+h\delta_{\rm K})\Lambda_\eps}{\SB(D)}\right)
=
O\!\left(\frac{\eps}{\SB(D)}\right).
\]
Since $P_{\rm emp}$ has nonnegative empirical GFT, the resulting mixture $M^\dagger$ satisfies
\[
\Rev(M^\dagger;\widehat D)\ge \rho_{\rm add},
\qquad
\GFT(M^\dagger;\widehat D)\ge \SB(D)-O(\eps).
\]
\end{proof}

\begin{proof}[Proof of Theorem~\ref{thm:add-upper-main}]
By Lemma~\ref{lem:marginal-kolmogorov}, the event
$\Delta_{\rm K}(\widehat D,D)\le\delta_{\rm K}$ holds with probability at least $1-\gamma$ whenever
\[
m\ge \frac{1}{2\delta_{\rm K}^2}\log\frac4\gamma,
\]
which is implied by the stated sample bound since
$\delta_{\rm K}=c_{\rm K}\eps/(h\Lambda_\eps)$.

On this event, Lemma~\ref{lem:empirical-mixed-main} gives an empirical feasible mechanism with GFT at least $\SB(D)-O(\eps)$.  Therefore the empirical-margin LP is feasible, and its optimum has empirical GFT at least $\SB(D)-O(\eps)$.  By Lemma~\ref{lem:lp}, the mechanism $\widehat M$ returned by the algorithm has the same empirical GFT and revenue as an optimal LP solution, so
\[
\GFT(\widehat M;\widehat D)\ge \SB(D)-O(\eps),
\qquad
\Rev(\widehat M;\widehat D)\ge \rho_{\rm add}.
\]
Lemma~\ref{lem:add-transfer-main} gives
\[
\GFT(\widehat M;D)\ge \SB(D)-O(\eps)-O(h\delta_{\rm K}),
\qquad
\Rev(\widehat M;D)\ge \rho_{\rm add}-O(h\delta_{\rm K}).
\]
The constants in $\rho_{\rm add}$ and $\delta_{\rm K}$ make the true revenue nonnegative and the total GFT loss at most $\eps$.  \Cref{lem:fin-impl-main,lem:grid-step-extension} give BIC and IIR, and nonnegative true revenue gives ex-ante WBB.
\end{proof}

\subsection{Proof of Lemma~\ref{lem:additive-kl}}

\begin{proof}
The seller marginal is common, so it suffices to compare the buyer marginals.  For sufficiently small $\eta$, assume $\eta\|r\|_\infty\le 1/2$.  Then
\[
\mathrm{KL}(D_B^+\|D_B^-)
=
\int_0^1
(1+\eta r(z))
\log\frac{1+\eta r(z)}{1-\eta r(z)}\,dz .
\]
For $|x|\le 1/2$,
\[
(1+x)\log\frac{1+x}{1-x}=2x+O(x^2).
\]
Since $\int_0^1 r(z)\,dz=0$, this gives
\[
\mathrm{KL}(D_B^+\|D_B^-)
\le
C\eta^2\int_0^1 r(z)^2\,dz
=
O(\eta^2).
\]
The same argument with the two signs exchanged gives
\[
\mathrm{KL}(D_B^-\|D_B^+)=O(\eta^2).
\]
Because $D^+$ and $D^-$ differ only in the buyer marginal,
\[
\mathrm{KL}(D^+\|D^-),\ \mathrm{KL}(D^-\|D^+)=O(\eta^2).
\]
\end{proof}

\subsection{Proof of Lemma~\ref{lem:additive-cross-world-separation}}

\begin{proof}
For an allocation $q$, write $a(z,y):=q(hz,hy)$ and set
\[
J(z):=\int_z^1 r(t)\,dt .
\]
Under $D^\theta$,
\[
\GFT(q;D^\theta)
=
h\int_{[0,1]^2} a(z,y)(z-y)\bigl(1+\theta\eta r(z)\bigr)\,dzdy,
\]
and the virtual-surplus expression is
\[
h\int_{[0,1]^2} a(z,y)K^\theta(z,y)\,dzdy,
\]
where
\[
K^\theta(z,y)
=
2z-1-2y+\theta\eta\bigl((z-2y)r(z)-J(z)\bigr).
\]
Here $K^\theta$ is the virtual-surplus density after multiplying by the buyer density
factor $1+\theta\eta r(z)$; it is not the pointwise virtual surplus alone.  Let
\[
T:=\{(z,y):z-y\ge 1/4\},\qquad
G_0:=\int_T (z-y)\,dzdy .
\]
For the threshold allocation $1_T$,
\[
\GFT(1_T;D^+)
=
h\left(G_0+\eta L_{\rm sep}\right),
\qquad
L_{\rm sep}:=\int_T (z-y)r(z)\,dzdy .
\]
Its revenue equals its virtual surplus, namely
\[
\eta h B_{\rm sep},
\qquad
B_{\rm sep}:=\int_T\bigl((z-2y)r(z)-J(z)\bigr)\,dzdy .
\]
Integrating first in $y$ and then by parts gives
\[
B_{\rm sep}
=
\frac12\int_{1/4}^1
\left(z-\frac14\right)\left(\frac34-z\right)r(z)\,dz>0,
\]
by the chosen supports of $r_H$ and $r_L$.  Hence $1_T$ is ex-ante WBB under $D^+$ for
$\eta>0$, and
\[
\SB(D^+)\ge h\left(G_0+\eta L_{\rm sep}\right).
\]

Now let $M$ be BIC, IIR, and ex-ante WBB under $D^-$, with allocation $q$. Lemma~\ref{lem:vs-ineq-main} and $\Rev(M;D^-)\ge0$ imply 
\[
\int_{[0,1]^2}a(z,y)K^-(z,y)\,dzdy\ge0.
\]
Therefore
\[
\begin{aligned}
\frac1h\GFT(M;D^+)
&\le
\int a(z,y)\left[(z-y)(1+\eta r(z))+\frac12K^-(z,y)\right]\,dzdy  \\
&\le
\int_{[0,1]^2}
\left[
2(z-y)-\frac12+\frac{\eta}{2}\bigl(zr(z)+J(z)\bigr)
\right]_+ dzdy .
\end{aligned}
\]
Since the boundary of $T$ is a line segment and $zr(z)+J(z)$ is bounded,
\[
\int
\left[
2(z-y)-\frac12+\frac{\eta}{2}\bigl(zr(z)+J(z)\bigr)
\right]_+ dzdy
=
G_0+\eta U_{\rm sep}+O(\eta^2),
\]
where
\[
U_{\rm sep}:=\frac12\int_T\bigl(zr(z)+J(z)\bigr)\,dzdy .
\]
A direct calculation using the displayed formula for $B_{\rm sep}$ gives
\[
L_{\rm sep}-U_{\rm sep}=\frac12B_{\rm sep}>0.
\]
Thus, for all sufficiently small $\eta$,
\[
\GFT(M;D^+)
\le
h\left(G_0+\eta U_{\rm sep}+O(\eta^2)\right)
\le
\SB(D^+)-c\eta h
\]
for a universal constant $c>0$.  This proves the lemma.
\end{proof}

\begin{lemma}\label{lem:two-point-disjoint}
Let $D^+=D_B^+\times D_S^+$ and $D^-=D_B^-\times D_S^-$.  Let
$\mathcal G^+$ and $\mathcal G^-$ be disjoint mechanism classes.  If a learner using
$m$ independent samples from each marginal returns a mechanism in $\mathcal G^\theta$ with probability at least $2/3$ under $D^\theta$, for each $\theta\in\{+,-\}$, then
\[
m\left(
\KL(D_B^+\|D_B^-)+\KL(D_S^+\|D_S^-)
\right)=\Omega(1).
\]
\end{lemma}

\begin{proof}
Let $P_+$ and $P_-$ be the joint laws of the learner's samples and internal randomness under $D^+$ and $D^-$, respectively.  The event
\[
A:=\{\text{the learner returns a mechanism in }\mathcal G^+\}
\]
satisfies $P_+(A)\ge2/3$.  Since $\mathcal G^+\cap\mathcal G^-=\varnothing$, success under $D^-$ implies $P_-(A)\le1/3$.  Hence
\[
\TV(P_+,P_-)\ge |P_+(A)-P_-(A)|\ge \frac13 .
\]
Pinsker's inequality and product additivity of Kullback--Leibler divergence \citep[Lemma~2.5(i); Section~2.4, property~(iii)]{Tsybakov2009} give
\[ \frac19 \le \frac12\KL(P_+\|P_-) = \frac m2\left( \KL(D_B^+\|D_B^-)+\KL(D_S^+\|D_S^-) \right).
\]
This proves the claim.
\end{proof}

\subsection{Proof of Theorem~\ref{thm:add-lower-main}}

\begin{proof}
\xhdr{Regularity of the additive hard pair.}
For $z=b/h$, set
\[ J(z):=\int_z^1 r(t)\,dt . \]
Then
\[ 1-F_B^\theta(hz)=1-z+\theta\eta J(z), \qquad h f_B^\theta(hz)=1+\theta\eta r(z). \]
Since $r\in C_c^2((0,1))$ and $\int_0^1 r(z)\,dz=0$, both $r$ and $J$ vanish in neighborhoods of the endpoints.  For sufficiently small $\eta$, the densities are positive.  Moreover,
\[ \frac{1-F_B^\theta(hz)}{h f_B^\theta(hz)} = \frac{1-z+\theta\eta J(z)}{1+\theta\eta r(z)}
\]
has derivative $-1+O(\eta)$, uniformly over $z\in[0,1]$. Hence the buyer inverse hazard rate is decreasing, so the buyer marginal is regular. The hazard rate is
\[ \frac{f_B^\theta(hz)}{1-F_B^\theta(hz)} = \frac1h\frac{1+\theta\eta r(z)}{1-z+\theta\eta J(z)} . \]
On the compact region containing the perturbation, the hazard is a $C^1$-small perturbation of the uniform hazard $(h(1-z))^{-1}$, whose derivative is bounded away from zero there.  Near the endpoints it equals the uniform hazard.  Thus the buyer marginal is MHR for sufficiently small $\eta$.  The seller marginal is uniform, hence regular and MHR, with nondecreasing virtual cost.

By the regularity verification above, $D^+$ and $D^-$ are regular product distributions.  Let
\[
\mathcal G^\theta := \left\{ M:\ M \text{ is BIC, IIR, and ex-ante WBB under }D^\theta,\  \GFT(M;D^\theta)\ge \SB(D^\theta)-\eps \right\}. \]
By Lemma~\ref{lem:additive-cross-world-separation}, if $M$ is ex-ante WBB under $D^-$, then
\[
\GFT(M;D^+)\le \SB(D^+)-c\eta h .
\]
Thus $\mathcal G^+\cap\mathcal G^-=\varnothing$ whenever $\eps<c\eta h$.

Set $\eta:=2\eps/(ch)$, decreasing the constant $c_1$ in the theorem so that $\eta$ lies in the perturbative range required above and in Lemma~\ref{lem:additive-cross-world-separation}.  Since the seller marginal is common, Lemmas~\ref{lem:two-point-disjoint} and~\ref{lem:additive-kl} give
\[
m=\Omega(\eta^{-2})
=
\Omega\!\left(\frac{h^2}{\eps^2}\right).
\]
This proves the theorem.
\end{proof}
\section{Proofs for Section~\ref{sec:main-bounded-mult}}\label{sec:bsm}

\subsection{Proof of Lemma~\ref{lem:bounded-mult-scale}}

\begin{proof}
Let $Z_i=(b_i^0-s_i^0)^+$.  Then $0\le Z_i\le h$, $\E Z_i=\FB(D)$, and
\[
\Var(Z_i)\le \E Z_i^2\le h\E Z_i=h\FB(D).
\]
By Bernstein's inequality,
\[
\Prob\left(\left|\widehat V-\FB(D)\right|\ge \frac12\FB(D)\right)
\le 2\exp\left(-c\,\frac{m_0\FB(D)}{h}\right)
\]
for a universal constant $c>0$.  The stated pilot size gives
\[
\frac12\FB(D)\le \widehat V\le \frac32\FB(D)
\]
with probability at least $1-\gamma$.  The bounds
\[
\SB(D)\le \FB(D)\le 3.15 \cdot \SB(D)
\]
from Lemma~\ref{lem:scale-main} give the final claim.
\end{proof}

\subsection{Proof of Lemma~\ref{lem:localized-transfer}}

\begin{lemma}[{\citep[Corollary~4]{Panchenko2003}}\footnote{This formulation is a corollary by plugging the half line indicators as the function class $\mathcal{F}$ in \citep[Corollary~4]{Panchenko2003}}]\label{lem:panchenko-vc}
Let $P$ be a distribution on $[0,h]$, and let $\widehat P$ be the empirical distribution
formed from independent samples $X_1,\ldots,X_m\sim P$.  Let
\[
\mathcal H:=\{[a,h]:0\le a\le h\}\cup\{[0,a]:0\le a\le h\}.
\]
For $H\in\mathcal H$, write
\[
Q_m(H):=P(H)-\widehat P(H),
\]
and
\[
V_H
:=4\left(
\operatorname{Var}_P(\one_H)
+\operatorname{Var}_{\widehat P}(\one_H)
+Q_m(H)^2
\right),
\]
where
\[
\operatorname{Var}_{\widehat P}(\one_H)
:=\frac1m\sum_{i=1}^m\left(\one_H(X_i)-\widehat P(H)\right)^2 .
\]
For every fixed $\beta\in(0,1)$, there is a constant $K_\beta>0$ such that, for every
$t\ge \log(1/\beta)$,
\[
\Pr\left(
\exists H\in\mathcal H:
\frac{|Q_m(H)|}{\sqrt{V_H}}
\ge
K_\beta\sqrt{\frac{\log m}{m}}+\sqrt{\frac{t}{m}}
\right)
\le
\exp\left\{1-\left(\sqrt t-\sqrt{\log(1/\beta)}\right)^2\right\}.
\]
\end{lemma}

\begin{lemma}\label{lem:panchenko-halfline}
Let $P$ be a distribution on $[0,h]$, and let $\widehat P$ be the empirical distribution
formed from $m$ independent samples from $P$.  Let
\[
\mathcal H:=\{[t,h]:0\le t\le h\}\cup\{[0,t]:0\le t\le h\}.
\]
There is a universal constant $C>0$ such that, with probability at least $1-\gamma$,
simultaneously for every $H\in\mathcal H$,
\[
\bigl|\widehat P(H)-P(H)\bigr|
\le
C\left(
\sqrt{\frac{(P(H)+\widehat P(H))\ell}{m}}
+
\frac{\ell}{m}
\right),
\qquad
\ell:=\log\!\left(\frac{Cm}{\gamma}\right).
\]
\end{lemma}

\begin{proof}
Apply Lemma~\ref{lem:panchenko-vc}, with $\beta=1/2$ and
$t=C\log(C/\gamma)$.  For $f=\one_H$,
\[
\operatorname{Var}_P(f)\le P(H),\qquad
\operatorname{Var}_{\widehat P}(f)\le \widehat P(H).
\]
Thus, on an event of probability at least $1-\gamma$,
\[
x_H:=|P(H)-\widehat P(H)|
\le
C\sqrt{\frac{(P(H)+\widehat P(H)+x_H^2)\ell}{m}}
\]
for all $H\in\mathcal H$, after enlarging $C$.  Solving this scalar inequality gives
\[
x_H
\le
C\left(
\sqrt{\frac{(P(H)+\widehat P(H))\ell}{m}}
+
\frac{\ell}{m}
\right),
\]
where the case $\ell/m$ larger than an absolute constant is covered by increasing $C$ since
$x_H\le1$.  This proves the lemma.
\end{proof}

\begin{lemma}\label{lem:oned-monotone}
Under the event of Lemma~\ref{lem:panchenko-halfline}, every monotone function
$u:[0,h]\to[0,h]$ satisfies
\[
\left|\E_{\widehat P}u-\E_Pu\right|
\le
C\left(
\sqrt{\frac{h(\E_Pu+\E_{\widehat P}u)\ell}{m}}
+
\frac{h\ell}{m}
\right).
\]
\end{lemma}

\begin{proof}
Assume first that $u$ is nondecreasing.  By the layer-cake representation,
\[
u(x)=\int_0^h \one\{u(x)\ge t\}\,dt .
\]
For each $t$, the set $\{x:u(x)\ge t\}$ is a half-line in $\mathcal H$.  Therefore Lemma~\ref{lem:panchenko-halfline}
and Cauchy's inequality give
\[
\begin{aligned}
\left|\E_{\widehat P}u-\E_Pu\right|
&\le
\int_0^h
\left|\widehat P(u\ge t)-P(u\ge t)\right|\,dt\\
&\le
C\int_0^h
\left(
\sqrt{\frac{(P(u\ge t)+\widehat P(u\ge t))\ell}{m}}
+
\frac{\ell}{m}
\right)dt\\
&\le
C\left(
\sqrt{\frac{h\ell}{m}
\int_0^h(P(u\ge t)+\widehat P(u\ge t))\,dt}
+
\frac{h\ell}{m}
\right)\\
&=
C\left(
\sqrt{\frac{h(\E_Pu+\E_{\widehat P}u)\ell}{m}}
+
\frac{h\ell}{m}
\right).
\end{aligned}
\]
The nonincreasing case is identical, using lower half-lines.
\end{proof}

\begin{lemma}\label{lem:rel-bound}
Let $\widehat D=\widehat D_B\times\widehat D_S$ be the empirical product distribution from
$m$ samples from each marginal.  With probability at least $1-\gamma$, every measurable
$f:[0,h]^2\to[0,h]$ nondecreasing in $b$ and nonincreasing in $s$ satisfies
\[
\left|\E_{\widehat D}f-\E_Df\right|
\le
C\left(
\sqrt{\frac{h(\E_Df+\E_{\widehat D}f)\ell}{m}}
+
\frac{h\ell}{m}
\right),
\qquad
\ell:=\log\!\left(\frac{C m }{\gamma}\right).
\]
\end{lemma}

\begin{proof}
Let
\[
\begin{aligned}
\mathcal M_+&:=\{\phi:[0,h]\to[0,h]:\phi\text{ is nondecreasing}\},\\
\mathcal M_-&:=\{\psi:[0,h]\to[0,h]:\psi\text{ is nonincreasing}\}.
\end{aligned}
\]
By Lemma~\ref{lem:oned-monotone}, after absorbing the union bound into the constant in $\ell$, with probability at least $1-\gamma$ the following two events hold simultaneously:
\[
\mathcal E_B:=
\left\{
\forall \phi\in\mathcal M_+,\quad
\left|\E_{\widehat D_B}\phi-\E_{D_B}\phi\right|
\le
C\left(\sqrt{\frac{h(\E_{D_B}\phi+\E_{\widehat D_B}\phi)\ell}{m}}
+\frac{h\ell}{m}\right)
\right\},
\]
and
\[
\mathcal E_S:=
\left\{
\forall \psi\in\mathcal M_-,\quad
\left|\E_{\widehat D_S}\psi-\E_{D_S}\psi\right|
\le
C\left(\sqrt{\frac{h(\E_{D_S}\psi+\E_{\widehat D_S}\psi)\ell}{m}}
+\frac{h\ell}{m}\right)
\right\}.
\]
Work on $\mathcal E_B\cap\mathcal E_S$ and fix any admissible $f$.  Write
\[
P_0:=\E_{D_B\times D_S}f,\qquad
T:=\E_{D_B\times\widehat D_S}f,\qquad
Q:=\E_{\widehat D_B\times\widehat D_S}f,\qquad
k:=\frac{h\ell}{m}.
\]

First define the seller-coordinate projection
\[
u_f(s):=\E_{B\sim D_B}f(B,s).
\]
Since $f$ is nonincreasing in its seller coordinate, $u_f\in\mathcal M_-$.  Moreover,
\[
\E_{D_S}u_f
=
\E_{S\sim D_S}\E_{B\sim D_B}f(B,S)
=
\E_{D_B\times D_S}f
=P_0,
\]
while
\[
\E_{\widehat D_S}u_f
=
\E_{S\sim\widehat D_S}\E_{B\sim D_B}f(B,S)
=
\E_{D_B\times\widehat D_S}f
=T.
\]
Applying $\mathcal E_S$ to $u_f$ gives
\[
|T-P_0|\le C(\sqrt{k(P_0+T)}+k).
\]

Next fix the realized seller sample, and define the buyer-coordinate projection
\[
v_f(b):=\E_{S\sim\widehat D_S}f(b,S)
=\frac1m\sum_{j=1}^m f(b,s_j).
\]
For this fixed seller sample, $v_f$ is a deterministic element of $\mathcal M_+$ because
$f$ is nondecreasing in its buyer coordinate.  The event $\mathcal E_B$ is simultaneous over all functions in $\mathcal M_+$, so it applies to this particular $v_f$.  The two expectations of $v_f$ are
\[
\E_{D_B}v_f
=
\E_{B\sim D_B}\E_{S\sim\widehat D_S}f(B,S)
=
\E_{D_B\times\widehat D_S}f
=T,
\]
and
\[
\E_{\widehat D_B}v_f
=
\E_{B\sim\widehat D_B}\E_{S\sim\widehat D_S}f(B,S)
=
\E_{\widehat D_B\times\widehat D_S}f
=Q.
\]
Applying $\mathcal E_B$ to $v_f$ gives
\[
|Q-T|\le C(\sqrt{k(Q+T)}+k).
\]

It remains to remove the hybrid term $T$ from the right-hand side.  From the first display,
\[
T\le P_0+C\sqrt{k(P_0+T)}+Ck.
\]
Using $\sqrt{k(P_0+T)}\le \sqrt{kP_0}+\sqrt{kT}$ and AM--GM,
\[
C\sqrt{kP_0}\le \frac12 P_0+Ck,
\qquad
C\sqrt{kT}\le \frac12 T+Ck,
\]
after increasing $C$.  Hence $T\le C(P_0+k)$.  The same argument applied to
\[
T\le Q+C\sqrt{k(Q+T)}+Ck
\]
gives $T\le C(Q+k)$.  Substituting $T\le C(P_0+k)$ into the seller-coordinate bound and $T\le C(Q+k)$ into the buyer-coordinate bound yields
\[
|T-P_0|\le C(\sqrt{kP_0}+k),
\qquad
|Q-T|\le C(\sqrt{kQ}+k).
\]
Finally, by the triangle inequality and $\sqrt{kP_0}+\sqrt{kQ}\le \sqrt{2k(P_0+Q)}$,
\[
|Q-P_0|\le C(\sqrt{k(P_0+Q)}+k),
\]
as claimed.
\end{proof}



\begin{proof}[Proof of Lemma~\ref{lem:localized-transfer}]
It suffices to write the deterministic kernels obtained after fixing the report-independent random seed and then average these kernels over the randomization.  Let $\Theta$ be the seed, independent of the reports.  Conditional on $\Theta=\theta$, the mechanism is a deterministic threshold component with allocation $q^\theta$ and threshold maps $\tau_B^\theta,\tau_S^\theta$.  Define three kernels
\[
\begin{aligned}
g_M(b,s)&:=\E_\Theta\bigl[q^\Theta(b,s)(b-s)\bigr],\\
u^B_M(b,s)&:=\E_\Theta\bigl[q^\Theta(b,s)(b-\tau_B^\Theta(s))\bigr],\\
u^S_M(b,s)&:=\E_\Theta\bigl[q^\Theta(b,s)(\tau_S^\Theta(b)-s)\bigr].
\end{aligned}
\]
For any $D_0\in\{D,\Db\}$,
\[
\begin{aligned}
\GFT(M;D_0)&=\E_{D_0}g_M,\\
U(M;D_0)&=\E_{D_0}u^B_M+\E_{D_0}u^S_M,\\
\Rev(M;D_0)&=\GFT(M;D_0)-U(M;D_0).
\end{aligned}
\]
Since $q(b,s)=0$ whenever $b<s$, we have $q^\theta(b,s)=0$ whenever $b<s$ for almost every seed $\theta$.  Hence the three deterministic kernels inside the $\Theta$-expectations are nonnegative, $[0,h]$-bounded, nondecreasing in $b$, and nonincreasing in $s$ for almost every $\theta$.  Averaging over $\Theta$ preserves these properties, so $g_M,u^B_M,u^S_M$ also have them.

We prove the case in which $M$ is $A$-localized under $\Db$; the case in which
$M$ is $A$-localized under $D$ is the same after interchanging $D$ and $\Db$.
The localization assumption under $\Db$ gives
\[
\E_{\Db}g_M\le A,\qquad
\E_{\Db}u^B_M\le A,\qquad
\E_{\Db}u^S_M\le A.
\]
Indeed, the first inequality is $\GFT(M;\Db)\le A$, while the last two follow from nonnegativity and
$\E_{\Db}u^B_M+\E_{\Db}u^S_M=U(M;\Db)\le A$.

Let $f$ be any one of $g_M,u^B_M,u^S_M$.  Lemma~\ref{lem:rel-bound} gives
\[
\left|\E_D f-\E_{\Db}f\right|
\le
C\left(
\sqrt{\frac{h\ell}{m}\bigl(\E_D f+\E_{\Db}f\bigr)}
+\frac{h\ell}{m}
\right).
\]
For the three kernels above, the localization inequalities give $\E_{\Db}f\le A$.  Since
\[
\E_D f\le \E_{\Db}f+\left|\E_D f-\E_{\Db}f\right|,
\]
we have
\[
\E_D f+\E_{\Db}f
\le
2A+\left|\E_D f-\E_{\Db}f\right|.
\]
Therefore
\[
\left|\E_D f-\E_{\Db}f\right|
\le
C\left(
\sqrt{\frac{hA\ell}{m}}
+
\sqrt{\frac{h\ell}{m}\left|\E_D f-\E_{\Db}f\right|}
+
\frac{h\ell}{m}
\right).
\]
Using AM--GM to absorb the middle term,
\[
C\sqrt{\frac{h\ell}{m}\left|\E_D f-\E_{\Db}f\right|}
\le
\frac12\left|\E_D f-\E_{\Db}f\right|+\frac{C^2}{2}\frac{h\ell}{m}.
\]
Substituting this bound into the previous display and moving the half-error term to the left gives
\[
\left|\E_D f-\E_{\Db}f\right|
\le
2C\sqrt{\frac{hA\ell}{m}}+(2C+C^2)\frac{h\ell}{m} \leq O\left(\sqrt{\frac{hA\ell}{m}}+\frac{h\ell}{m}\right)
\]

Thus the same bound applies to each of $g_M,u_M^B,u_M^S$.  In particular,
\[
\left|\GFT(M;D)-\GFT(M;\Db)\right|\le  O\left(\sqrt{\frac{hA\ell}{m}}+\frac{h\ell}{m}\right),
\]
and
\[
\begin{aligned}
\left|U(M;D)-U(M;\Db)\right|
&\le
\left|\E_Du^B_M-\E_{\Db}u^B_M\right|
+
\left|\E_Du^S_M-\E_{\Db}u^S_M\right|\leq  O\left(\sqrt{\frac{hA\ell}{m}}+\frac{h\ell}{m}\right).
\end{aligned}
\]
The revenue bound follows from $\Rev=\GFT-U$.


\end{proof}

\subsection{Proof of Theorem~\ref{thm:bounded-mult-upper-main}}
\begin{proof}
Apply Lemma~\ref{lem:bounded-mult-scale} and Lemma~\ref{lem:localized-transfer} with failure probability $\gamma/2$ each.  By a union bound, the two events hold simultaneously with probability at least $1-\gamma$.  Replacing $\gamma$ by $\gamma/2$ only changes the logarithmic factor by an absolute constant, so we keep the notation $\ell$.  On the pilot event,
\[
\frac12\FB(D)\le \widehat V\le \frac32\FB(D),
\qquad
\SB(D)\le \FB(D)\le 3.15\,\SB(D),
\]
and hence $\widehat V=\Theta(\SB(D))$ and
\[
\widehat\Lambda:=1+\log\frac h{\widehat V}=\Theta(\Lambda_h(D)).
\]

Let $C_{\rm loc}$ be a universal constant large enough for the $O(\cdot)$ bounds in Lemma~\ref{lem:localized-transfer}.  Fix $C_0\ge4$, then choose $c_\rho>0$ sufficiently small, and finally choose the hidden universal constant in the training sample size sufficiently large.  More explicitly, after $C_0$ and $c_\rho$ are fixed, the stated training sample bound implies
\[
m\ge
K_m\,\frac{h\ell}{\widehat V}\frac{\widehat\Lambda^2}{\alpha^2}
\]
for a universal constant $K_m$ that we may choose as large as needed.  Since
\[
\rho=c_\rho\frac{\alpha\widehat V}{\widehat\Lambda},
\]
choosing $K_m$ large enough gives
\begin{equation}\label{eq:bounded-transfer-small}
C_{\rm loc}\left(
\sqrt{\frac{hC_0\widehat V\ell}{m}}
+
\frac{h\ell}{m}
\right)
\le
\frac{\rho}{16}.
\end{equation}
Indeed, the sample-size lower bound gives
\[
C_{\rm loc}\sqrt{\frac{hC_0\widehat V\ell}{m}}
\le
C_{\rm loc}\sqrt{\frac{C_0}{K_m}}\frac{\alpha\widehat V}{\widehat\Lambda}
\le
\frac{\rho}{32},
\]
and, since $\alpha<1/2$ and $\widehat\Lambda\ge1$,
\[
C_{\rm loc}\frac{h\ell}{m}
\le
\frac{C_{\rm loc}}{K_m}\frac{\alpha^2\widehat V}{\widehat\Lambda^2}
\le
\frac{C_{\rm loc}}{K_m}\frac{\alpha\widehat V}{\widehat\Lambda}
\le
\frac{\rho}{32}.
\]

Let $M^\ast$ be a second-best canonical mechanism for $D$.  By surplus trimming, we may assume that it trades only on $\{b\ge s\}$, without decreasing GFT or revenue.  Indeed, for a deterministic threshold component with allocation $q$ and threshold maps $\tau_B,\tau_S$, replacing
\[
q(b,s)\quad\text{by}\quad q(b,s)\one\{b\ge s\}
\]
keeps the allocation threshold implementable, with threshold maps
\[
\widetilde\tau_B(s):=\max\{\tau_B(s),s\},
\qquad
\widetilde\tau_S(b):=\min\{\tau_S(b),b\}.
\]
On $b<s$, regularity gives $\phi_B(b)\le b<s\le \phi_S(s)$, so deleting such trades weakly increases both GFT and the virtual-surplus expression for revenue.  The same operation is applied componentwise to randomized threshold mechanisms.

Lemma~\ref{lem:spread-main} gives a posted-spread mechanism $P$ with
\[
\Rev(P;D)=\Omega\!\left(\frac{\SB(D)}{\Lambda_h(D)}\right).
\]
Posted-spread mechanisms are also surplus-trimmed.  Applying Lemma~\ref{lem:mixing} to the surplus-trimmed $M^\ast$ and $P$ with target margin $8\rho$ gives a surplus-trimmed mechanism $W$ such that
\[
\Rev(W;D)\ge 8\rho,
\qquad
\GFT(W;D)\ge (1-O(c_\rho\alpha))\SB(D),
\]
where the mixture weight is $O(\rho\Lambda_h(D)/\SB(D))=O(c_\rho\alpha)$ and is at most one by the choice of $c_\rho$.  Moreover, $W$ is $C_0\widehat V$-localized under $D$, since
\[
\GFT(W;D)\le \FB(D)\le 2\widehat V\le \frac{C_0}{2}\widehat V,
\qquad
\GFT(W;D)-\Rev(W;D)\le \GFT(W;D)\le \frac{C_0}{2}\widehat V,
\]
where the last inequalities use $C_0\ge4$.
Lemma~\ref{lem:localized-transfer}, applied to $W$, gives
\[
\Rev(W;\widehat D)\ge 8\rho-\frac{\rho}{16}\ge 4\rho,
\qquad
\GFT(W;\widehat D)\le C_0\widehat V.
\]
Indeed, the revenue inequality uses Lemma~\ref{lem:localized-transfer} and \eqref{eq:bounded-transfer-small}.  For the GFT cap, the same transfer bound gives
\[
\GFT(W;\widehat D)
\le
\GFT(W;D)+\frac{\rho}{16}
\le
\frac{C_0}{2}\widehat V+\frac{\rho}{16}
\le
C_0\widehat V,
\]
where the last step uses $\rho=c_\rho\alpha\widehat V/\widehat\Lambda$, $\alpha<1/2$, $\widehat\Lambda\ge1$, and the fixed choice of sufficiently small $c_\rho$.  Let $W_{\rm emp}:=\Pi_{\widehat D}W$.  By Lemma~\ref{lem:projection},
\[
\GFT(W_{\rm emp};\widehat D)=\GFT(W;\widehat D),
\qquad
\Rev(W_{\rm emp};\widehat D)\ge \Rev(W;\widehat D).
\]
Also, $W_{\rm emp}$ satisfies the empirical negative-surplus constraint: on every empirical grid profile with $\beta_i<\sigma_j$, the allocation of $W$ is zero because $W$ is surplus-trimmed, and the projection preserves the grid allocation.  Therefore $W_{\rm emp}$ is feasible for the localized empirical LP.  In particular, on the current event the LP is feasible, so Algorithm~\ref{alg:mult-upper-main} does not return no trade.

Let $q^\star$ be an optimal solution of the localized empirical LP.  Since
$W_{\rm emp}$ is feasible, optimality gives the only comparison needed here:
\[
\widehat{\GFT}(q^\star)
\ge
\GFT(W_{\rm emp};\widehat D).
\]
It remains only to justify that this LP optimum is implemented by the mechanism class used in the transfer lemma.  After imposing the zero-allocation constraints $q_{ij}=0$ for $\beta_i<\sigma_j$, the remaining monotone allocation set is a face of the finite-grid order polytope.  Intersecting this face with the empirical revenue halfspace and the empirical GFT-cap halfspace, the same face-dimension argument as in Lemma~\ref{lem:lp} shows that an optimal LP solution may be chosen in a face of dimension at most two.  By Carath\'eodory's theorem inside this face, it is a convex combination of at most three monotone integral grid allocations.  The finite-grid implementation therefore realizes $q^\star$ as a randomized threshold mechanism $\widehat M$, with the same empirical GFT and revenue.  In particular,
\[
\GFT(\widehat M;\widehat D)
=
\widehat{\GFT}(q^\star)
\ge
\GFT(W_{\rm emp};\widehat D).
\]

The LP constraints give $\Rev(\widehat M;\widehat D)\ge4\rho$ and $\GFT(\widehat M;\widehat D)\le C_0\widehat V$.  Hence
\[
U(\widehat M;\widehat D)
=\GFT(\widehat M;\widehat D)-\Rev(\widehat M;\widehat D)
\le \GFT(\widehat M;\widehat D)
\le C_0\widehat V.
\]
Thus $\widehat M$ is $C_0\widehat V$-localized under $\widehat D$.

We also need the surplus-trimming hypothesis of Lemma~\ref{lem:localized-transfer} for $\widehat M$.  Each deterministic component of the stepwise extension uses
\[
\pi_B(b)=\max\{\beta_i:\beta_i\le b\},
\qquad
\pi_S(s)=\min\{\sigma_j:\sigma_j\ge s\}.
\]
If $b<s$, then $\pi_B(b)\le b<s\le\pi_S(s)$, so $\pi_B(b)<\pi_S(s)$.  The grid constraint $q_{ij}=0$ for $\beta_i<\sigma_j$ therefore implies that every deterministic component, and hence $\widehat M$, has zero allocation whenever $b<s$.

Then Lemma~\ref{lem:localized-transfer} gives
\[
\GFT(\widehat M;D)
\ge
\GFT(\widehat M;\widehat D)-\frac{\rho}{16}
\ge
\GFT(W_{\rm emp};\widehat D)-\frac{\rho}{16}.
\]
By Lemma~\ref{lem:projection} and another application of Lemma~\ref{lem:localized-transfer} to
$W$,
\[
\GFT(W_{\rm emp};\widehat D)-\frac{\rho}{16}
=
\GFT(W;\widehat D)-\frac{\rho}{16}
\ge
\GFT(W;D)-\frac{\rho}{16}-\frac{\rho}{16}
\ge
(1-\alpha)\SB(D),
\]
after choosing the universal constant $c_{\rho}$ small enough. Similarly,
\[
\Rev(\widehat M;D)
\ge
\Rev(\widehat M;\widehat D)-\frac{\rho}{16}
\ge
4\rho-\frac{\rho}{16}
\ge0.
\]
Thus $\widehat M$ is ex-ante WBB under $D$.  Lemmas~\ref{lem:fin-impl-main} and
\ref{lem:grid-step-extension} give BIC and IIR for each deterministic component, and
report-independent randomization preserves both properties. The pilot and training sample sizes are bounded by the stated
\[
\Oe\!\left(
\frac{h}{\SB(D)}
\frac{\Lambda_h(D)^2}{\alpha^2}
\log\frac1\gamma
\right),
\]
using Lemmas~\ref{lem:bounded-mult-scale} and~\ref{lem:scale-main}.
\end{proof}

\subsection{Proof of Theorem~\ref{thm:bounded-mult-lower-main}}

\begin{proof}
Let $H^\theta_R=X^\theta_R\times Y_R$, $\theta\in\{+,-\}$, be the scale-$R$ additive hard pair from Subsection~\ref{ssec:add-lower-main}.  By Lemmas~\ref{lem:additive-kl} and~\ref{lem:additive-cross-world-separation},
\[
\KL(H^+_R\|H^-_R),\ \KL(H^-_R\|H^+_R)=O(\eta^2),
\qquad
\SB(H^\pm_R)=\Theta(R),
\]
and the cross-world WBB separation is $\Omega(\eta R)$.

Embed this pair in a probability-$p$ active buyer block.  Let $x_\theta(u):=F^{-1}_{X^\theta_R}(1-u)$ be the buyer upper-tail quantile in the additive pair, and define the upper-tail quantile of the embedded buyer by
\[
b_\theta(q):=
\begin{cases}
R+x_\theta(q/p), & 0<q\le p,\\[1mm]
pR/q, & p<q\le 1.
\end{cases}
\]
Equivalently, the buyer CDF is
\[
F_{B^\theta_p}(b)=
\begin{cases}
1-pR/b, & pR\le b\le R,\\[1mm]
1-p\bigl(1-F_{X_R^\theta}(b-R)\bigr), & R\le b\le 2R .
\end{cases}
\]
Let the seller be $S=R+Y_R$, and denote the resulting product distribution by $D^\theta_p$. 
The two branches join at $q=p$, since $b_\theta(p)=R+x_\theta(1)=R$ and $pR/p=R$.  Hence the buyer support is contained in $[pR,2R]$, while the seller support is contained in $[R,2R]$.

To verify buyer regularity, consider the buyer revenue curve $\mathcal R_\theta(q):=q b_\theta(q)$ in the upper-tail quantile parameter $q$.  It is flat on the inactive branch, $\mathcal R_\theta(q)=pR$ for $q>p$, and on the active branch it is $\mathcal R_\theta(q)=q(R+x_\theta(q/p))$.  
The active-branch curve is concave because $X_R^\theta$ is regular and adding the linear term $qR$ preserves concavity.  
At the joining point, the perturbation is supported away from the endpoint, so $x_\theta(1)=0$ and $x_\theta'(1)=-R$; hence the active and inactive pieces have matching one-sided derivatives, both equal to zero, at $q=p$. 
Thus $\mathcal R_\theta$ is concave across the join, and the buyer marginal is regular.  
The seller marginal is the translate $S=R+Y_R$, so seller regularity is preserved.

On the inactive block, buyer values are at most $R$ while seller values are at least $R$, so realized surplus is nonpositive.  The inactive buyer virtual value is zero, because $b_\theta(q)+q b_\theta'(q)=0$ for $q>p$, while the seller virtual cost is at least $R$.  Hence the inactive block also has nonpositive virtual surplus and cannot contribute to the benchmark or offset the cross-world WBB separation.  Since the inactive block has nonpositive realized surplus,
\[
\SB(D_p^\theta)
\le
p\,\FB(H_R^\theta)
=
O(pR).
\]
Conversely, take a feasible canonical mechanism for $H_R^\theta$, shift its buyer and seller reports and payments upward by $R$, run it on the active buyer block, and use no trade on the inactive block. Seller interim utilities and expected revenue are multiplied by $p$. Buyer incentive compatibility is preserved within each block, and an inactive buyer cannot profitably report an active type because every active buyer critical payment is at least $R$. Hence the extended mechanism is feasible on $D_p^\theta$ and
\[
\SB(D_p^\theta)
\ge
p\,\SB(H_R^\theta)
=
\Omega(pR).
\]
Therefore
\[
\SB(D_p^\theta)=\Theta(pR).
\]

Now let $M$ be BIC, IIR, and ex-ante WBB under $D_p^-$, with allocation rule $q_M$. By the envelope inequality,
\[
\E_{D_p^-}\!\left[
q_M(B,S)\bigl(\phi_B^-(B)-\phi_S(S)\bigr)
\right]
\ge0.
\]
Define the shifted active allocation
\[
\widetilde q(x,y):=q_M(R+x,R+y).
\]
The inactive-region virtual surplus is nonpositive, so the active-region virtual-surplus integral is nonnegative. After shifting the active reports downward by $R$ and substituting $q=pu$, the active-region GFT and virtual-surplus expressions are exactly $p$ times the corresponding expressions in the additive hard pair. The variational argument in the proof of Lemma~\ref{lem:additive-cross-world-separation} therefore gives
\[
\GFT(M;D_p^+)
\le
p\,\SB(H_R^+)-cp\eta R.
\]
Since the inactive-region GFT is nonpositive and
\[
p\,\SB(H_R^+)\le\SB(D_p^+),
\]
we obtain
\[
\GFT(M;D_p^+)
\le
\SB(D_p^+)-cp\eta R.
\]
Thus the cross-world WBB separation is $\Omega(p\eta R)$.

Finally, the two buyer marginals coincide on the inactive block and, conditional on the active block, are the buyer marginals of the additive hard pair. Since the seller marginal is common,
\[
\KL(D_p^+\|D_p^-)
=
p\,\KL(H_R^+\|H_R^-)
=
O(p\eta^2),
\]
and similarly
\[
\KL(D_p^-\|D_p^+)=O(p\eta^2).
\]

Set $R=h/2$ and $\eta=c_\alpha\alpha$, where $c_\alpha>0$ is an absolute
constant chosen within the perturbative range so that the separation exceeds the
$\alpha\SB(D^\pm_p)$ tolerance.  Then $D^+_p,D^-_p$ are supported on
$[0,h]^2$, satisfy
\[
\SB(D^\pm_p)=\Theta(ph),
\qquad
\KL(D^+_p\|D^-_p),\ \KL(D^-_p\|D^+_p)=O(p\alpha^2),
\]
and their feasible $(1-\alpha)$-optimal classes are disjoint.  Lemma~\ref{lem:two-point-disjoint}
therefore gives
\[
m=\Omega\!\left(\frac1{p\alpha^2}\right).
\]
Using $\SB(D^\pm_p)=\Theta(ph)$, this becomes
\[
m=\Omega\!\left(\frac{h}{\SB(D^\pm_p)}\frac1{\alpha^2}\right).
\]
\end{proof}


\section{Proofs for Section~\ref{sec:main-mhr}}\label{sec:mhr}

\subsection{Algorithm and supporting lemmas}

We use the notation $D_w$ and $\Ext_w(\cdot)$ from Section~\ref{sec:main-mhr}. For each marginal, partition the samples into two independent blocks: a cap-selection block and a bounded-learner block. Split the cap-selection block into independent subsamples for mean estimation and the posted-spread search, and split the bounded-learner block into the pilot and training subsamples required by Algorithm~\ref{alg:mult-upper-main}.

Throughout this appendix, for every cap $u>0$, the seller sentinel $\top_S$ in $D_u$ is treated as a fixed no-trade state and is not assigned a numerical cost. A mechanism on $D_u$ is specified on the numerical report region $[0,u]^2$ and extended to the sentinel row by
\[
q(b,\top_S)=x_B(b,\top_S)=x_S(b,\top_S)=0
\qquad
\text{for every }b\in[0,u].
\]
BIC and IIR are imposed on the numerical buyer and seller types under the full law $D_u$. The sentinel is also an admissible report for a numerical seller and yields zero utility. We use $\SB(D_u)$ for the supremum GFT over BIC, IIR, and ex-ante-WBB mechanisms in this class. Sentinel observations retain their probability mass. On such profiles, first-best surplus, GFT, buyer and seller utilities, revenue, and the pilot statistic are zero. In the empirical LP, sentinel observations form a fixed zero row, contribute zero to the empirical GFT, utility, and revenue expressions, and are excluded from the numerical seller envelope-payment formula. We identify no trade with the zero mechanism.

Fix a sufficiently large universal constant $C_{\rm cap}$. Given the mean estimate $\widehat\mu$, the posted-spread mechanism $P$, and its lower-confidence revenue $\widehat R>0$ produced by the cap-selection stage, set
\[
\widehat w
:=
C_{\rm cap}\widehat\mu
\left(
1+\log\frac{2\widehat\mu}{\alpha\widehat R}
\right).
\]
Algorithm~\ref{alg:mhr-main} projects the bounded-learner samples to the truncated instance at width $\widehat w$, applies Algorithm~\ref{alg:mult-upper-main}, and extends the resulting mechanism back to the original instance.

\begin{algorithm}[htbp]
\caption{MHR estimated-cap learner}\label{alg:mhr-main}
\KwIn{Accuracy $\alpha$, confidence $\gamma$, and independent samples from each marginal}
\KwOut{A mechanism for the original instance, or no trade}
Compute the mean estimate $\widehat\mu$ from the mean-estimation split of the cap-selection samples\;
Run the posted-spread search on the remaining cap-selection samples, obtaining a posted-spread mechanism $P$ and its lower-confidence revenue $\widehat R$\;
\If{$\widehat R=0$}{
    \Return no trade\;
}
Project the bounded-learner pilot and training samples to the truncated instance $D_{\widehat w}$\;
Run Algorithm~\ref{alg:mult-upper-main} on the projected instance $D_{\widehat w}$ with $h=\widehat w$, accuracy $\alpha/2$, and confidence $\gamma/3$, treating $\widehat w$ as a buyer endpoint and adjoining the fixed zero seller-sentinel row; denote its output by $N_{\widehat w}$\;
\Return $\Ext_{\widehat w}(N_{\widehat w})$\;
\end{algorithm}

\begin{lemma}\label{lem:mhr-tail-main}
There is a universal constant $C>0$ such that, for every nonnegative MHR random variable $X$ with mean $\mu_X>0$ and every $v\ge0$,
\[
\E[(X-v)^+]\le C\mu_X e^{-v/\mu_X},
\qquad
v\Pr[X>v]\le C(v+\mu_X)e^{-v/\mu_X}.
\]
Consequently, for
\[
\Xi_v(D)
:=
\E[(B-v)^+]
+
v\Pr[B>v],
\]
one has
\[
\Xi_v(D)
\le
C(v+\mu(D))
\exp\!\left(-\frac{v}{\mu(D)}\right).
\]
\end{lemma}

\begin{proof}
Let $\bar F(t):=\Pr[X>t]$. By the MHR characterization in \citep[Section~3]{BMP1963}, $\log\bar F$ is concave on the support. The case $v=0$ is immediate. Fix $v>0$ and set $p:=\bar F(v)$. The cases $p\in\{0,1\}$ are immediate. For $p\in(0,1)$, set $y:=-\log p$. By concavity, for $0\le t\le v$,
\[
\bar F(t)\ge p^{t/v}=e^{-yt/v}.
\]
Thus
\[
\mu_X=\int_0^\infty \bar F(t)\,dt
\ge
\int_0^v e^{-yt/v}\,dt
=
v\frac{1-e^{-y}}{y}.
\]
Hence $v/\mu_X\le y/(1-e^{-y})\le y+1$, and therefore
\[
\Pr[X>v]=\bar F(v)\le e\exp(-v/\mu_X).
\]
Moreover, concavity of $\log\bar F$ and $\bar F(0)=1$ imply
\[
\bar F(v+t)\le \bar F(v)\bar F(t),
\qquad t\ge0.
\]
Therefore
\[
\E[(X-v)^+]
=
\int_0^\infty \bar F(v+t)\,dt
\le
\bar F(v)\mu_X
\le
e\mu_X\exp(-v/\mu_X).
\]
The same survival bound gives
\[
v\Pr[X>v]\le e(v+\mu_X)\exp(-v/\mu_X).
\]
Applying these estimates to $B$ and using $\E[B]\le\mu(D)$ gives the displayed bound on $\Xi_v(D)$ after increasing the universal constant.
\end{proof}

Fix a sufficiently large universal constant $C_{\rm ps}$. Given a mean estimate $\widehat\mu$ computed from an independent sample, let $m_{\rm ps}$ be the number of posted-spread search samples and set
\[
\ell_{\rm ps}
:=
\log\!\left(
\frac{C_{\rm ps}m_{\rm ps}}{\gamma}
\right),
\qquad
H_0
:=
C_{\rm ps}\widehat\mu
\left(
\ell_{\rm ps}+\log\frac1\alpha
\right).
\]
Pair the buyer and seller search samples independently, and denote the resulting pairs by $(b_i,s_i)_{i=1}^{m_{\rm ps}}$. For $0\le x\le y\le H_0$, define
\[
r(x,y)
:=
(y-x)\Pr[S\le x,\ B\ge y],
\qquad
\widehat r(x,y)
:=
\frac{y-x}{m_{\rm ps}}
\sum_{i=1}^{m_{\rm ps}}
\one\{s_i\le x,\ b_i\ge y\}.
\]
Choose $(\widehat x,\widehat y)$ to maximize
\[
\left[
\widehat r(x,y)
-
C_{\rm ps}
\left(
\sqrt{
\frac{(y-x)\widehat r(x,y)\ell_{\rm ps}}
{m_{\rm ps}}
}
+
\frac{(y-x)\ell_{\rm ps}}
{m_{\rm ps}}
\right)
\right]_+
\]
over $0\le x\le y\le H_0$. Let $P$ be the posted-spread mechanism with prices $\widehat x,\widehat y$, and let $\widehat R$ be the attained maximum.
The prices may be chosen among the search sample values in $[0,H_0]$ together with $0$ and $H_0$. The posted-spread search therefore has at most $O(m_{\rm ps}^2)$ candidate price pairs and can be implemented by direct enumeration in polynomial time. Projecting the bounded-learner samples adds only the buyer endpoint at $\widehat w$ and a fixed zero seller-sentinel row; the latter does not enter the numerical seller envelope-payment grid. Hence the bounded-learning stage remains a polynomial-size finite LP, and Algorithm~\ref{alg:mhr-main} runs in polynomial time.

\begin{lemma}\label{lem:posted-spread-search}
Suppose that $\widehat\mu$ is computed from an independent sample and satisfies
\[
\frac12\mu(D)
\le
\widehat\mu
\le
2\mu(D).
\]
If
\[
m_{\rm ps}
\ge
C_{\rm ps}
\chi_\mu(D)
L_{\mu,\alpha}(D)^2
\ell_{\rm ps},
\]
then, with probability at least $1-\gamma/12$ over the posted-spread search samples,
\[
c_R
\frac{\SB(D)}
{L_{\mu,\alpha}(D)}
\le
\widehat R
\le
\Rev(P;D)
\le
\SB(D)
\]
for a universal constant $c_R>0$.
\end{lemma}

\begin{proof}
For fixed $x,y$, the revenue sample is bounded by $y-x$ and has variance at most $(y-x)r(x,y)$. Since the class of events $\{s\le x,\ b\ge y\}$ has constant VC dimension, the relative VC inequality~\citep[Corollary~4]{Panchenko2003}, together with the standard empirical-form conversion, gives, by the choice of $C_{\rm ps}$, with probability at least $1-\gamma/12$, simultaneously for all $0\le x\le y\le H_0$,
\[
\left|\widehat r(x,y)-r(x,y)\right|
\le
C_{\rm ps}
\left(
\sqrt{
\frac{(y-x)r(x,y)\ell_{\rm ps}}
{m_{\rm ps}}
}
+
\frac{(y-x)\ell_{\rm ps}}
{m_{\rm ps}}
\right),
\]
and
\[
r(x,y)
\ge
\widehat r(x,y)
-
C_{\rm ps}
\left(
\sqrt{
\frac{(y-x)\widehat r(x,y)\ell_{\rm ps}}
{m_{\rm ps}}
}
+
\frac{(y-x)\ell_{\rm ps}}
{m_{\rm ps}}
\right).
\]
Applying this inequality at $(\widehat x,\widehat y)$ gives $\widehat R\le\Rev(P;D)$. The mechanism $P$ is BIC, IIR, and ex-post WBB because $\widehat x\le\widehat y$, and its realized GFT is at least its realized revenue. Hence
\[
\widehat R\le\Rev(P;D)\le\GFT(P;D)\le\SB(D).
\]
For a fixed empirical trade set, the lower-confidence score is linear in $y-x$, so a maximizer may be chosen at the search sample values or at $0,H_0$.

It remains to prove the lower bound on $\widehat R$. Let $h^\star$ be a sufficiently large universal multiple of
$\mu(D)L_{\mu,\alpha}(D)$, and let $\overline D_{h^\star}$ be the product distribution of $\min\{B,h^\star\}$ and $\min\{S,h^\star\}$. Lemma~\ref{lem:mhr-tail-main} gives
\[
\E[(B-h^\star)^+]\le c\frac{\SB(D)}{L_{\mu,\alpha}(D)}
\]
for a sufficiently small universal constant $c>0$. Since
\[
\FB(\overline D_{h^\star})\ge\FB(D)-\E[(B-h^\star)^+],
\]
we have $\FB(\overline D_{h^\star})\ge\SB(D)/2$ by the choice of $h^\star$.

Lemma~\ref{lem:spread-main} applied to $\overline D_{h^\star}$ gives prices $0\le x^\star\le y^\star\le h^\star$ such that
\[
r(x^\star,y^\star)\ge c\frac{\SB(D)}{L_{\mu,\alpha}(D)}.
\]
Indeed,
\[
\frac{h^\star}{\FB(\overline D_{h^\star})}=O\!\left(\chi_\mu(D)L_{\mu,\alpha}(D)\right),
\]
so the logarithmic denominator in Lemma~\ref{lem:spread-main} is $O(L_{\mu,\alpha}(D))$. Clipping at $h^\star$ does not change the revenue of these prices: if $x^\star<y^\star$, then $x^\star<h^\star$, while if $x^\star=y^\star$, the revenue is zero. Thus the same lower bound holds under $D$.

The sample condition gives $m_{\rm ps}\ge C\chi_\mu(D)$, and hence
\[
\ell_{\rm ps}+\log\frac1\alpha\ge c\left(1+\log\frac{\chi_\mu(D)}\alpha\right)=cL_{\mu,\alpha}(D).
\]
Together with $\widehat\mu\ge\mu(D)/2$, this implies $H_0\ge h^\star$ by the choice of $C_{\rm ps}$. Moreover,
\[
\frac{y^\star-x^\star}{r(x^\star,y^\star)}=O\!\left(\chi_\mu(D)L_{\mu,\alpha}(D)^2\right),
\]
so the stated sample bound gives
\[
\frac{(y^\star-x^\star)\ell_{\rm ps}}{m_{\rm ps}}\le c\,r(x^\star,y^\star)
\]
for a sufficiently small universal constant $c>0$. On the relative VC event, this implies that $\widehat r(x^\star,y^\star)$ is within a fixed fraction of $r(x^\star,y^\star)$ and that the lower-confidence subtraction at $(x^\star,y^\star)$ is at most another fixed fraction of $r(x^\star,y^\star)$. The maximizing choice of $(\widehat x,\widehat y)$ therefore gives
\[
\widehat R\ge c\,r(x^\star,y^\star)\ge c_R\frac{\SB(D)}{L_{\mu,\alpha}(D)}.
\]
\end{proof}

The cap-selection stage is used only through the following event.

\begin{lemma}\label{lem:mhr-pilot-main}
With
\[
\Oe\!\left(
\frac{
\chi_\mu(D)
L_{\mu,\alpha}(D)^3
}{\alpha^2}
\log\frac1\gamma
\right)
\]
cap-selection samples from each marginal, the cap-selection stage returns a mean estimate $\widehat\mu$, a posted-spread mechanism $P$, and a lower-confidence revenue $\widehat R$ such that, with probability at least $1-\gamma/3$, the bounds
\begin{equation}\label{eq:mhr-pilot-event}
\frac12\mu(D)
\le
\widehat\mu
\le
2\mu(D),
\qquad
c_R
\frac{\SB(D)}
{L_{\mu,\alpha}(D)}
\le
\widehat R
\le
\Rev(P;D)
\le
\SB(D)
\end{equation}
hold, where $c_R>0$ is the universal constant from Lemma~\ref{lem:posted-spread-search}.
\end{lemma}

\begin{proof}
Use independent sample splits for mean estimation and posted-spread search. Split the mean-estimation samples into $O(\log(1/\gamma))$ blocks, each of sufficiently large constant size, and let $\widehat\mu$ be the sum of the medians of the buyer and seller block averages. By Lemma~\ref{lem:mhr-tail-main}, for $X\in\{B,S\}$,
\[
\E[X^2] = 2\int_0^\infty v\Pr[X>v]\,dv = O(\E[X]^2).
\]
Hence a sufficiently large constant block size gives a constant-factor estimate of the corresponding mean with constant probability. Taking the median over the independent blocks and applying a union bound to the two marginals gives
\[
\frac12\mu(D)
\le
\widehat\mu
\le
2\mu(D)
\]
with probability at least $1-\gamma/12$. This proves the first part of \eqref{eq:mhr-pilot-event}.

Condition on this event. Since the posted-spread search uses an independent sample split, Lemma~\ref{lem:posted-spread-search} applies to the fixed estimate $\widehat\mu$. Allocating a fixed fraction of the cap-selection samples to the posted-spread search, the stated sample bound implies
\[
m_{\rm ps}
\ge
C_{\rm ps}
\chi_\mu(D)
L_{\mu,\alpha}(D)^2
\ell_{\rm ps}.
\] 
Therefore, with conditional probability at least $1-\gamma/12$,
\[
c_R\frac{\SB(D)}{L_{\mu,\alpha}(D)}\le\widehat R\le\Rev(P;D)\le\SB(D).
\]
A union bound gives \eqref{eq:mhr-pilot-event} with probability at least $1-\gamma/6$, which implies the stated probability.
\end{proof}

\begin{lemma}\label{lem:mhr-cap-compare}
Let $M$ be a surplus-trimmed randomized canonical threshold mechanism on $D$. Let $M_v$ be the mechanism obtained by applying the same threshold rule after projecting to $D_v$, with the fixed zero outcome on the seller-sentinel row. Then
\[
\GFT(M_v;D_v)
\ge
\GFT(M;D)-2\Xi_v(D),
\]
and
\[
\Rev(M_v;D_v)
\ge
\Rev(M;D)-2\Xi_v(D).
\]
The same bounds hold when restricting the finite truncated instance $D_H$ further down to $D_v$, $v<H$.
\end{lemma}

\begin{proof}
It suffices to prove the result for a deterministic component of $M$ and then average over the report-independent randomization. Let $\tau_B$ and $\tau_S$ be its buyer and seller threshold maps. Couple $D$ and $D_v$ by capping the buyer type at $v$ and sending every seller type $S>v$ to the seller sentinel.

We bound the GFT and revenue losses under the same profile partition. On every profile traded by the original mechanism, canonical threshold payments and individual rationality give
\[
\tau_B(S)\le B,
\qquad
\tau_S(B)\ge S,
\]
and therefore
\[
\bigl(\tau_B(S)-\tau_S(B)\bigr)^+
\le(B-S)^+.
\]

If $B\le v$ and $S\le v$, the allocation is unchanged. The GFT contribution is therefore unchanged, while restricting the type space cannot increase the seller critical payment and leaves the buyer critical payment of every retained trade unchanged. Hence revenue does not decrease.

Suppose that $S\le v<B$. If the restricted mechanism trades at $(v,S)$, monotonicity implies that the original mechanism trades at $(B,S)$. The GFT loss is $B-v\le B$, while revenue does not decrease because the buyer critical payment is unchanged and
\[
\tau_S(v)\le\tau_S(B).
\]
If the restricted mechanism does not trade, any deleted GFT contribution is at most $B-S\le B$, and any positive deleted revenue contribution is at most
\[
\bigl(\tau_B(S)-\tau_S(B)\bigr)^+ \le B-S \le B.
\]

Finally, if $S>v$, the restricted mechanism does not trade. Since $M$ is surplus-trimmed, an original trade requires $B\ge S>v$. Both its GFT contribution and its positive revenue contribution are then at most
\[
B-S\le B-v\le(B-v)^+.
\]

Thus, for each of the two objectives, the pointwise loss is bounded above by
\[
B\mathbf 1\{S\le v<B\} + (B-v)^+\mathbf 1\{S>v\} \le B\mathbf 1\{B>v\} + (B-v)^+.
\]
Taking expectations and using
\[
\E[B\mathbf 1\{B>v\}] = v\Pr[B>v]+\E[(B-v)^+],
\]
the expected loss in either objective is at most
\[
v\Pr[B>v]+2\E[(B-v)^+] \le 2\Xi_v(D).
\]
Taking expectations proves both stated bounds.

For a finite truncated instance $D_H$, couple $D_H$ and $D_v$ through the same original types $(B,S)$. Restricting first at $H$ and then at $v$ gives the same projected reports as restricting directly at $v$. Moreover, the capped buyer type is at most $B$, and every traded active seller type above $v$ still satisfies $B\ge S>v$ by surplus trimming. Hence the same pointwise bounds apply, and the additional GFT and revenue losses are bounded by $2\Xi_v(D)$. Averaging over the deterministic components completes the proof.
\end{proof}

\begin{lemma}\label{lem:conservative-main}
Let $N$ be a randomized canonical threshold mechanism on the truncated instance $D_w$, with the fixed zero outcome on the seller-sentinel row and canonical payments on the numerical report region. Then $\Ext_w(N)$ is BIC and IIR on $D$, and
\[
\Rev(\Ext_w(N);D)=\Rev(N;D_w),
\qquad
\GFT(\Ext_w(N);D)\ge\GFT(N;D_w).
\]
In particular, if $N$ is ex-ante WBB under $D_w$, then $\Ext_w(N)$ is ex-ante WBB under $D$.
\end{lemma}

\begin{proof}
The extension runs $N$ after mapping $(B,S)$ to $(B_w,S_w)$. For buyer types below the cap, incentives are inherited from $N$. For buyer types above the cap, reporting the capped type is optimal because the allocation is monotone and the canonical buyer payment is the critical-value payment at the capped type.

For a seller type $s\le w$, incentives among the numerical reports are inherited from $N$. Reporting the sentinel gives zero utility, which is no larger than the truthful utility by IIR. For a seller type $s>w$, truthful reporting leads to the fixed zero sentinel outcome. Any deviation to a numerical seller report can result in trade only at a canonical seller critical payment of at most $w$. Hence the deviating utility is at most
\[
x_S-sq\le(w-s)q\le0,
\]
while truthful reporting gives zero utility. Thus no seller type has a profitable deviation, and IIR follows from the same critical-payment representation.

Payments depend only on the projected reports, and the sentinel outcome has zero transfers. Their distribution under $D$ is therefore the same as under $D_w$, which gives
\[
\Rev(\Ext_w(N);D)=\Rev(N;D_w).
\]
Profiles with $S>w$ have zero allocation. On profiles with $S\le w$, the projected surplus agrees with the true surplus when $B\le w$ and is weakly smaller than the true surplus when $B>w$. Therefore
\[
\GFT(\Ext_w(N);D)\ge\GFT(N;D_w).
\]
\end{proof}

\begin{lemma}\label{lem:mhr-cap-preservation}
There is a universal constant $C>0$ such that the following holds for every $\alpha\in(0,1)$. Let $P$ be a posted-spread mechanism satisfying
\[
\Rev(P;D)\ge R_{\rm ps},
\qquad
0<R_{\rm ps}\le\SB(D).
\]
If
\[
w
\ge
C\mu(D)
\left(
1+\log\frac{2\mu(D)}{\alpha R_{\rm ps}}
\right),
\]
then
\[
\SB(D_w)
\ge
\left(1-\frac{\alpha}{2}\right)\SB(D).
\]
\end{lemma}

\begin{proof}
By Lemma~\ref{lem:mhr-tail-main}, the stated lower bound on $w$, with the universal constant $C$ sufficiently large, gives
\[
\Xi_w(D)
\le
\frac{\alpha R_{\rm ps}}{64}.
\]

Fix a sufficiently large $H>w$ such that both prices of $P$ lie below $H$, and fix an accuracy parameter $\zeta>0$. Let $D_H$ be the truncated instance at level $H$.

Choose a BIC, IIR, and ex-ante-WBB mechanism $M$ on $D$, with allocation rule $q$, such that
\[
\GFT(M;D)\ge\SB(D)-\zeta.
\]
By the envelope inequality,
\[
\E_D\!\left[q(B,S)\bigl(\phi_B(B)-\phi_S(S)\bigr)\right]\ge0.
\]
Under the definition of $D_H$ above, let $\mathcal A_H$ be the set of measurable allocations $a\in[0,1]$ with $a(b,\top_S)=0$. For $a\in\mathcal A_H$, let $\mathcal G_H(a)$ denote its GFT and let $\mathcal R_H(a)$ denote its envelope revenue, whose coefficient on the numerical seller region is
\[
r_H(b,s):=
\begin{cases}
\phi_B(b)-\phi_S(s),&b<H,\\
H-\phi_S(s),&b=H.
\end{cases}
\]
For the mixed continuous--atomic buyer marginal of $D_H$, the envelope formula gives
\[
\Rev(M';D_H)\le \mathcal R_H(a)
\]
for every BIC and IIR mechanism $M'$ with allocation $a$, with equality under the canonical normalized payments used below. At the buyer atom $H$, the buyer-payment coefficient is $H$, which gives the displayed coefficient $H-\phi_S(s)$.

Consider the single-constraint relaxation
\[
V_H:=\sup\bigl\{\mathcal G_H(a):a\in\mathcal A_H,\ \mathcal R_H(a)\ge0\bigr\}.
\]
Every BIC, IIR, and ex-ante-WBB mechanism counted by $\SB(D_H)$ induces a feasible allocation, so
\[
V_H\ge\SB(D_H).
\]

Let $P_H$ be the restriction of $P$ to $D_H$, and let $a_{P_H}$ be its allocation. By the choice of $H$,
\[
\mathcal R_H(a_{P_H})
=
\Rev(P_H;D_H)
=
\Rev(P;D)
\ge R_{\rm ps}>0.
\]
Thus the relaxation is strictly feasible. The attainable pairs
\[
\bigl\{(\mathcal R_H(a),\mathcal G_H(a)):a\in\mathcal A_H\bigr\}
\]
form a compact convex subset of $\mathbb R^2$, so convex strong duality gives
\[
V_H
=
\min_{\lambda\ge0}
\sup_{a\in\mathcal A_H}
\bigl\{\mathcal G_H(a)+\lambda\mathcal R_H(a)\bigr\}.
\]
The dual objective is at least
\[
\mathcal G_H(a_{P_H})+\lambda R_{\rm ps},
\]
and therefore diverges as $\lambda\to\infty$. Hence the minimum is attained by a finite multiplier $\lambda^\star\ge0$, and
\[
V_H
=
\sup_{a\in\mathcal A_H}
\bigl\{\mathcal G_H(a)+\lambda^\star\mathcal R_H(a)\bigr\}.
\]
For a numerical profile, the corresponding pointwise score is
\[
\ell_\lambda(b,s):=
\begin{cases}
\bigl(b+\lambda\phi_B(b)\bigr)-\bigl(s+\lambda\phi_S(s)\bigr),&b<H,\\
(1+\lambda)H-\bigl(s+\lambda\phi_S(s)\bigr),&b=H.
\end{cases}
\]
Regularity makes $\ell_\lambda$ nondecreasing in $b$ and nonincreasing in $s$. Moreover,
\[
b+\lambda\phi_B(b)\le(1+\lambda)b\le(1+\lambda)H,
\qquad b<H,
\]
so this monotonicity extends to the buyer atom. Hence
\[
a^-_\lambda:=\one\{\ell_\lambda>0\},
\qquad
a^+_\lambda:=\one\{\ell_\lambda\ge0\},
\]
with zero allocation on the seller-sentinel row, are deterministic canonical threshold allocations maximizing the Lagrangian at $\lambda$.

A zero-score profile cannot have $b<s$, since regularity would give $r_H(b,s)\le b-s<0$ and hence $\ell_{\lambda^\star}(b,s)<0$. Thus, when $\lambda^\star>0$, one has $r_H\le0$ on the zero-score set, and the subgradient optimality condition gives
\[
\mathcal R_H(a^+_{\lambda^\star})
\le0
\le
\mathcal R_H(a^-_{\lambda^\star}).
\]
A report-independent mixture of these two threshold allocations can therefore be chosen with envelope revenue zero. Since both allocations maximize the same Lagrangian, this mixture has GFT $V_H$.

If $\lambda^\star=0$, boundary optimality gives
\[
\mathcal R_H(a^-_0)\ge0,
\]
and $a^-_0$ has GFT $V_H$. Thus, in either case, $V_H$ is attained by a randomized canonical threshold allocation with nonnegative envelope revenue. Its canonical payments give a BIC, IIR, and ex-ante-WBB mechanism on $D_H$, and therefore
\[
V_H\le\SB(D_H).
\]
Together with the reverse inequality,
\[
V_H=\SB(D_H).
\]
Applying Lemma~\ref{lem:conservative-main} to the randomized canonical threshold mechanism attaining $V_H$ gives a BIC, IIR, and ex-ante-WBB mechanism on $D$ with GFT at least $V_H=\SB(D_H)$. Therefore
\[
\SB(D_H)\le\SB(D).
\]

On $D_H$, define
\[
q_H(b,s):=q(b,s)\one\{b<H,\ s\le H\},
\]
with zero allocation at the buyer cap atom and on the seller-sentinel row. Since $\E[B+S]<\infty$ and
\[
\bigl(\phi_B(B)-\phi_S(S)\bigr)^+
\le
(B-S)^+
\le B,
\]
the envelope inequality for $M$ and dominated convergence give a deterministic sequence $\varepsilon_H\to0$ such that
\[
\mathcal G_H(q_H)
\ge
\SB(D)-\zeta-\varepsilon_H,
\qquad
\mathcal R_H(q_H)\ge-\varepsilon_H.
\]
Mixing $q_H$ with $a_{P_H}$ with weight
\[
\frac{\varepsilon_H}{R_{\rm ps}+\varepsilon_H}
\]
on $a_{P_H}$ gives a feasible allocation for the relaxation, because
\[
\left(
1-\frac{\varepsilon_H}{R_{\rm ps}+\varepsilon_H}
\right)(-\varepsilon_H)
+
\frac{\varepsilon_H}{R_{\rm ps}+\varepsilon_H}R_{\rm ps}
=0.
\]
Its GFT is at least $\SB(D)-\zeta-o_H(1)$, and therefore
\[
\SB(D_H)\ge\SB(D)-\zeta-o_H(1).
\]

Let $M_H$ be obtained by applying the componentwise surplus-trimming operation from the proof of Theorem~\ref{thm:bounded-mult-upper-main} to the randomized canonical threshold mechanism attaining
\[
V_H=\SB(D_H).
\]
This preserves its canonical threshold form and does not decrease GFT or revenue, so
\[
\Rev(M_H;D_H)\ge0,
\qquad
\GFT(M_H;D_H)\ge\SB(D_H).
\]

Restrict $M_H$ further from $H$ to $w$, and denote the resulting mechanism on $D_w$ by $M_{H,w}$. Lemma~\ref{lem:mhr-cap-compare} gives
\[
\GFT(M_{H,w};D_w)
\ge
\GFT(M_H;D_H)-2\Xi_w(D),
\]
and
\[
\Rev(M_{H,w};D_w)
\ge
-2\Xi_w(D).
\]

Let $P_w$ be the restriction of $P$ to $D_w$. The same lemma gives
\[
\Rev(P_w;D_w)
\ge
R_{\rm ps}-2\Xi_w(D)
\ge
\frac{R_{\rm ps}}2.
\]
Set
\[
\lambda
:=
\frac{4\Xi_w(D)}{R_{\rm ps}}.
\]
Since
\[
\Xi_w(D) \le \frac{\alpha R_{\rm ps}}{64},
\]
we have
\[
0 \le \lambda \le \frac{\alpha}{16} < 1.
\]
The revenue of the mixture
\[
(1-\lambda)M_{H,w}+\lambda P_w
\]
is at least
\[
-2(1-\lambda)\Xi_w(D)
+
\lambda\frac{R_{\rm ps}}2
\ge0.
\]
Thus the mixture is BIC, IIR, and ex-ante WBB under $D_w$.

The GFT of the mixture is therefore at least
\[
\begin{aligned}
(1-\lambda)
\bigl(
\SB(D_H)-\zeta-2\Xi_w(D)
\bigr)
&\ge
\SB(D_H)-\zeta-2\Xi_w(D)-\lambda\SB(D)\\
&\ge
\SB(D)-2\zeta-o_H(1)
-2\Xi_w(D)-\lambda\SB(D).
\end{aligned}
\]
Using $R_{\rm ps}\le\SB(D)$ and
\[
\Xi_w(D) \le \frac{\alpha R_{\rm ps}}{64},
\]
we have
\[
2\Xi_w(D) \le \frac{\alpha}{32}\SB(D), \qquad \lambda\SB(D) \le \frac{\alpha}{16}\SB(D).
\]
Consequently,
\[
\begin{aligned}
\SB(D_w)  &\ge \left( 1-\frac{\alpha}{2} \right)\SB(D) -2\zeta-o_H(1).
\end{aligned}
\]
Letting $H\to\infty$ and then $\zeta\to0^+$ proves the lemma.
\end{proof}

\begin{lemma}\label{lem:mhr-cap-compatible}
Fix $\alpha'\in(0,1/2)$, $\gamma'\in(0,1)$, and $w>0$. Let $D$ satisfy Assumption~\ref{ass:mhr-main}, and let $D_w$ be its truncated instance with $\SB(D_w)>0$. Run Algorithm~\ref{alg:mult-upper-main} on independent projected pilot and training samples from $D_w$, using width $w$, accuracy $\alpha'$, and confidence $\gamma'$, adding the buyer atom at $w$ as a grid point and adjoining the fixed zero seller-sentinel row. Then, with probability at least $1-\gamma'$, the output $N_w$ is BIC, IIR, and ex-ante WBB under $D_w$, and
\[
\GFT(N_w;D_w)\ge(1-\alpha')\SB(D_w),
\]
using
\[
\Oe\!\left(
\frac{w}{\SB(D_w)}
\frac{\bigl(1+\log(w/\SB(D_w))\bigr)^2}{(\alpha')^2}
\log\frac1{\gamma'}
\right)
\]
samples from each marginal, split into the pilot and training parts required by Algorithm~\ref{alg:mult-upper-main}.
\end{lemma}

\begin{proof}
Let
\[
a_w:=\Pr[S\le w],
\qquad
D_w^{\rm act}
:=
\mathcal L(B_w)\times\mathcal L(S\mid S\le w).
\]
Since $\SB(D_w)>0$, one has $a_w>0$. Restricting a mechanism on $D_w$ to the numerical seller reports gives a mechanism on $D_w^{\rm act}$, and conversely any mechanism on $D_w^{\rm act}$ extends to $D_w$ by assigning the fixed zero outcome to the sentinel. The additional sentinel report does not affect the numerical seller incentive constraints because it gives zero utility, while truthful utility is nonnegative by IIR. Buyer interim utilities, GFT, and revenue are multiplied by $a_w$. Hence
\[
\FB(D_w)
=
a_w\FB(D_w^{\rm act}),
\qquad
\SB(D_w)
=
a_w\SB(D_w^{\rm act}).
\]
Lemma~\ref{lem:scale-main} therefore gives
\[
\FB(D_w)\le 3.15\,\SB(D_w).
\]
Likewise, extending a posted-spread mechanism from $D_w^{\rm act}$ by the zero sentinel outcome multiplies its revenue by $a_w$. Lemma~\ref{lem:spread-main} consequently yields a posted-spread mechanism $P_w$ satisfying
\[
\Rev(P_w;D_w)
\ge
c_{\rm spr}
\frac{\FB(D_w)}
{1+\log\!\bigl(w/\FB(D_w)\bigr)}.
\]
Since $\SB(D_w)>0$, one has $\FB(D_w)>0$, so the displayed lower bound is strictly positive. Thus $P_w$ supplies a strictly feasible point for the allocation relaxation at $H=w$. Repeating the preceding allocation-relaxation argument with $H=w$ and $P_w$ in place of $P_H$ provides a surplus-trimmed randomized canonical threshold benchmark for $\SB(D_w)$ with the fixed zero sentinel outcome.
For the pilot stage, define
\[
Z(b,s):=
\begin{cases}
(b-s)^+,&s\in[0,w],\\
0,&s=\top_S.
\end{cases}
\]
Then
\[
0\le Z\le w,
\qquad
\E[Z]=\FB(D_w),
\qquad
\E[Z^2]\le w\FB(D_w).
\]
The Bernstein calculation in Lemma~\ref{lem:bounded-mult-scale} therefore gives a constant-factor estimate of $\FB(D_w)$, and hence of $\SB(D_w)$.

For the training stage, use the usual numerical buyer and seller grids in $[0,w]$, including the buyer atom at $w$. Sentinel observations retain their empirical mass but form a fixed zero row. They introduce no allocation or payment variables and contribute zero to the empirical GFT, utility, and revenue expressions. Thus the allocation variables, monotonicity constraints, canonical payments, and order-polytope implementation on the numerical grid are exactly those used by Algorithm~\ref{alg:mult-upper-main}.

It remains to check the localized transfer step. Extend each seller-coordinate kernel used in Lemma~\ref{lem:localized-transfer} to the terminal sentinel state by assigning value zero. The kernels are nonnegative and nonincreasing on $[0,w]$, so this extension remains nonincreasing. The corresponding upper- and lower-half-line classes on $[0,w]\cup\{\top_S\}$ still have constant VC dimension. Hence the one-dimensional relative comparison, and therefore the localized GFT, utility, and revenue bounds, hold with $h=w$.

The projected benchmark and the posted-spread mechanism above provide the same comparison and revenue-margin mechanisms used in the proof of Theorem~\ref{thm:bounded-mult-upper-main}. Applying its LP feasibility, optimality, and transfer-back argument with $D_w$, $w$, $\alpha'$, and $\gamma'$ gives
\[
\GFT(N_w;D_w)
\ge
(1-\alpha')\SB(D_w),
\qquad
\Rev(N_w;D_w)\ge0.
\]
The pilot and training sample bounds become
\[
\Oe\!\left(
\frac{w}{\SB(D_w)}
\frac{\bigl(1+\log(w/\SB(D_w))\bigr)^2}{(\alpha')^2}
\log\frac1{\gamma'}
\right).
\]
\end{proof}

\subsection{Proof of Theorem~\ref{thm:mhr-upper-main}}

\begin{proof}
Split the confidence parameter between the cap-selection stage and the bounded learner. On the event \eqref{eq:mhr-pilot-event}, $\widehat R>0$, so Algorithm~\ref{alg:mhr-main} does not return no trade at the cap-selection stage. Moreover, the constant-factor estimate of $\mu(D)$ implies that, for $C_{\rm cap}$ sufficiently large, the selected cap $\widehat w$ satisfies the hypothesis of Lemma~\ref{lem:mhr-cap-preservation} with $R_{\rm ps}=\widehat R$. Hence
\[
\SB(D_{\widehat w})
\ge
\left(1-\frac{\alpha}{2}\right)\SB(D)
\ge
\frac12\SB(D)>0.
\]

By \eqref{eq:mhr-pilot-event} and the definition of $\chi_\mu(D)$,
\[
1+\log\frac{2\widehat\mu}{\alpha\widehat R}
\le
1+\log\left(
\frac{4\chi_\mu(D)L_{\mu,\alpha}(D)}
{c_R\alpha}
\right)
=
O\!\bigl(L_{\mu,\alpha}(D)\bigr).
\]
Consequently,
\[
\begin{aligned}
\widehat w
&=
O\!\bigl(\mu(D)L_{\mu,\alpha}(D)\bigr),\\
\frac{\widehat w}{\SB(D_{\widehat w})}
&=
O\!\bigl(\chi_\mu(D)L_{\mu,\alpha}(D)\bigr),\\
1+\log\frac{\widehat w}{\SB(D_{\widehat w})}
&=
O\!\bigl(L_{\mu,\alpha}(D)\bigr).
\end{aligned}
\]

Apply Lemma~\ref{lem:mhr-cap-compatible} to $D_{\widehat w}$ with accuracy $\alpha/2$ and confidence $\gamma/3$. With probability at least $1-\gamma/3$, the bounded learner returns $N_{\widehat w}$ satisfying
\[
\GFT(N_{\widehat w};D_{\widehat w})
\ge
\left(1-\frac{\alpha}{2}\right)
\SB(D_{\widehat w}).
\]
Let $\widehat M:=\Ext_{\widehat w}(N_{\widehat w})$.  By Lemma~\ref{lem:conservative-main},
$\widehat M$ is BIC, IIR, and ex-ante WBB on $D$, and
\[
\GFT(\widehat M;D)
\ge
\GFT(N_{\widehat w};D_{\widehat w}).
\]
Combining the two bounds gives
\[
\GFT(\widehat M;D)
\ge
\left(1-\frac{\alpha}{2}\right)^2
\SB(D)
\ge
(1-\alpha)\SB(D).
\]
Substituting the preceding bounds into Lemma~\ref{lem:mhr-cap-compatible} gives the stated bounded-learner sample complexity. The cap-selection bound from Lemma~\ref{lem:mhr-pilot-main} is of the same order, and a union bound gives success probability at least $1-\gamma$.
\end{proof}

\subsection{Proof of Theorem~\ref{thm:mhr-lower-main}}

\begin{proof}
For a perturbation parameter $\eta>0$ to be chosen below, let $X^\theta_R\times Y_R$, $\theta\in\{+,-\}$, be the scale-$R$ additive hard pair from Section~\ref{ssec:add-lower-main}, where $R$ is its value scale. Let
\[
x_\theta(u):=F^{-1}_{X_R^\theta}(1-u),\qquad u\in(0,1),
\]
be the buyer upper-tail quantile in the additive pair.  The perturbation functions are compactly supported away from the endpoints, so the buyer marginal coincides with the uniform baseline near $u=1$; in particular,
\[
x_\theta(1)=0,\qquad x_\theta'(1)=-R.
\]

Fix $p\in(0,1/4)$ sufficiently small. We reuse the rare-block embedding from the proof of Theorem~\ref{thm:bounded-mult-lower-main}, replacing the inactive branch $pR/q$ by an exponential-tail branch and shifting the active block by $R\log(1/p)$:

\[
b_\theta(q):=
\begin{cases}
R\log(1/p)+x_\theta(q/p),&0<q\le p,\\[1mm]
R\log(1/q),&p<q\le 1.
\end{cases}
\]
Let the seller be
\[
S=R\log(1/p)+Y_R,
\]
and denote the resulting product distribution by $D^\theta$.
The two branches join at $q=p$ because $x_\theta(1)=0$. Their derivatives also agree:
\[
\frac{d}{dq}\bigl(R\log(1/p)+x_\theta(q/p)\bigr)\Big|_{q=p}
=
\frac{x_\theta'(1)}p
=
-\frac Rp
=
\frac{d}{dq}\bigl(R\log(1/q)\bigr)\Big|_{q=p}.
\]
The buyer marginal is MHR for sufficiently small $\eta$. On the inactive branch $q>p$, $b_\theta(q)=R\log(1/q)$, so the hazard rate is $1/R$.  On the active branch $q\le p$, the distribution is the shifted additive buyer marginal $X_R^\theta$, which is MHR by the verification in Appendix~\ref{sec:bsu}.  The two branches join without a downward hazard jump because $x_\theta(1)=0$ and $x_\theta'(1)=-R$ give the common hazard value $1/R$ at $q=p$.  The seller marginal is a translate of $Y_R$, hence is MHR with nondecreasing virtual cost.  Thus $D^+$ and $D^-$ satisfy Assumption~\ref{ass:mhr-main}.

The event $q\le p$ is the active block. Conditional on this block, the instance is exactly the additive hard pair shifted by $R\log(1/p)$. Outside the active block, buyer values are at most $R\log(1/p)$, while seller costs are at least $R\log(1/p)$, so realized surplus is nonpositive. The inactive buyer virtual value is
\[
b_\theta(q)+q b_\theta'(q)
=
R\log(1/q)-R
\le
R\log(1/p)-R,
\qquad q>p,
\]
whereas the seller virtual cost is at least $R\log(1/p)$. Hence the inactive region also has nonpositive virtual surplus and cannot contribute to the benchmark or offset the cross-world WBB separation.

The common shift by $R\log(1/p)$ preserves both the real-surplus and virtual-surplus functionals in the active block. Writing $q=pu$ on this block, one has
\[
b_\theta(pu)=R\log(1/p)+x_\theta(u),
\]
and
\[
b_\theta(pu)+pu\,b_\theta'(pu) = R\log(1/p)+x_\theta(u)+u x_\theta'(u).
\]
Since $Y_R$ is uniform on $[0,R]$, for $y\in[0,R]$ the shifted seller type is $R\log(1/p)+y$, and
\[
b_\theta(pu)-\bigl(R\log(1/p)+y\bigr) = x_\theta(u)-y.
\]
Moreover,
\[
\phi_S\!\left(R\log(1/p)+y\right) = R\log(1/p)+2y,
\]
so
\[
\bigl(b_\theta(pu)+pu\,b_\theta'(pu)\bigr) - \phi_S\!\left(R\log(1/p)+y\right)
= x_\theta(u)+u x_\theta'(u)-2y.
\]
Hence, for every allocation rule, the active-block GFT and virtual-surplus integral are exactly $p$ times the corresponding integrals in the additive hard pair. Therefore the benchmark-scaling and cross-world-separation argument from the proof of Theorem~\ref{thm:bounded-mult-lower-main} applies verbatim after replacing the shift $R$ by $R\log(1/p)$, and gives
\[
\SB(D^\theta)=\Theta(pR),
\qquad
\text{cross-world WBB separation }=\Omega(p\eta R),
\]
and
\[
\KL(D^+\|D^-)=O(p\eta^2),
\qquad
\KL(D^-\|D^+)=O(p\eta^2).
\]

It remains to express $p$ through $\chi_\mu$. The seller mean is $\Theta(R\log(1/p))$. The inactive buyer tail has mean
\[
\int_p^1 R\log(1/q)\,dq=O(R),
\]
and the active buyer block contributes
\[
p\bigl(R\log(1/p)+O(R)\bigr)
=
o\bigl(R\log(1/p)\bigr).
\]
Hence
\[
\mu(D^\theta)
=
\Theta\!\left(R\log\frac1p\right),
\qquad
\chi_\mu(D^\theta)
=
\frac{\mu(D^\theta)}{\SB(D^\theta)}
=
\Theta\!\left(\frac{\log(1/p)}p\right).
\]
It follows that
\[
\log\frac1p
=
\Theta\!\left(L_\chi(D^\theta)\right),
\qquad
\frac1p
=
\Theta\!\left(
\frac{\chi_\mu(D^\theta)}{L_\chi(D^\theta)}
\right).
\]

Choose $\eta=\Theta(\alpha)$, with the implicit constant sufficiently large for the cross-world WBB separation to exceed the $\alpha\SB(D^\theta)$ tolerance. Since the theorem concerns sufficiently small $\alpha$, this choice remains within the perturbative regime.  The feasible $(1-\alpha)$-optimal classes for $D^+$ and $D^-$ are then disjoint.  Lemma~\ref{lem:two-point-disjoint} and the one-sample KL bound $O(p\alpha^2)$ give
\[
m\,p\alpha^2=\Omega(1).
\]
Thus
\[ m = \Omega\!\left(\frac1{p\alpha^2}\right) = \Omega\!\left( \frac{\chi_\mu(D^\theta)}{L_\chi(D^\theta)\alpha^2} \right). \]
Choosing $p$ sufficiently small so that $\chi_\mu(D^\theta)=\Theta(\chi)$ gives the theorem.
\end{proof}


\end{document}